\documentclass[preprint]{elsarticle}

\usepackage{amsmath,amsfonts,amssymb,amsthm}
\usepackage[geometry]{ifsym}
\usepackage{mathrsfs,tikz,booktabs}
\usepackage[latin1]{inputenc}
\usepackage[ruled,vlined]{algorithm2e}
\usepackage[normalem]{ulem}
\usepackage{comment}

\newcommand{\HH}{\mathcal{H}}
\newcommand{\EE}{\mathcal{E}}
\newcommand{\CC}{\mathcal{C}}
\newcommand{\NP}{\mathcal{NP}}
\newcommand{\grd}{\gamma_{\rm gr}}
\newcommand{\la}{\langle}
\newcommand{\ra}{\rangle}
\newcommand{\antes}{\vartriangleleft}

\def\tick{\tikz\fill[scale=0.4](0,.35) -- (.25,0) -- (1,.7) -- (.25,.15) -- cycle;}

\newtheorem{thm}{Theorem}[section]
\newtheorem{prop}[thm]{Proposition}
\newtheorem{obs}[thm]{Observation}

\newtheorem{cor}[thm]{Corollary}

\journal{}

\begin{document}

\begin{frontmatter}

\title{An integer programming approach for solving a generalized version of the Grundy domination number\tnoteref{grant}}

\author[a]{Manoel Camp\^elo}
\author[b,c]{Daniel Sever\'in\fnref{correspon}}

\address[a]{Dep. Estat\'istica e Matem\'atica Aplicada, Universidade Federal do Cear\'a, Brazil}

\address[b]{Depto. de Matem\'atica (FCEIA), Universidad Nacional de Rosario, Argentina}

\address[c]{CONICET, Argentina}

\tnotetext[grant]{Partially supported by grants PICT-2016-0410 (ANPCyT), PID ING538 (UNR),
                                                443747/2014-8, 305264/2016-8 (CNPq) and PNE 0112­00061.01.00/16 (FUNCAP/CNPq).\\
\emph{E-mail addresses}: \texttt{mcampelo@lia.ufc.br} (M. Camp\^elo),
                         \texttt{daniel@fceia.unr.edu.ar} (D. Sever\'in).}
\fntext[correspon]{Corresponding author at Departamento de Matem\'atica (FCEIA), UNR, Pellegrini 250, Rosario, Argentina.}

\begin{abstract}
A \emph{legal dominating sequence} of a graph is an ordered dominating set of vertices where each element dominates at least another one not dominated by its predecessors in the sequence.
The length of a largest legal dominating sequence is called \emph{Grundy domination number}.
In this work, we introduce a generalized version of the Grundy domination problem.
We explicitly calculate the corresponding parameter for paths and web graphs. We propose integer programming formulations for the new
problem, find families of valid inequalities and perform extensive computational experiments to compare the formulations as well as to test these inequalities as cuts in a branch-and-cut framework. We also design and evaluate the performance of a heuristic for finding good initial lower and upper bounds and
a tabu search that improves the initial lower bound. The test instances include randomly generated graphs, structured graphs, classical benchmark instances and two instances from a real application. Our approach is exact for graphs with 20-50 vertices and provides good solutions for
graphs up to 10000 vertices.
\end{abstract}


\end{frontmatter}

\section{Introduction} \label{SSINTRO}

Covering problems are some of the most studied problems in graph theory and combinatorial optimization due to the
large number of applications.
Consider a hypergraph $\HH = (X, \EE)$ without isolated vertices.
An \emph{edge cover} of $\HH$ is a set of hyperedges $\CC \subseteq \EE$ that cover all vertices of $\HH$, i.e.$\!$ $\cup_{C \in \CC} C = X$.
The general covering problem consists in finding the \emph{covering number} of $\HH$ which is the minimum number of hyperedges in an edge cover of $\HH$ \cite{BERGE}.

The most natural constructive heuristic for obtaining an edge cover of $\HH$ is as follows.
Start from empty sets $\CC$ and $W$ (the latter one keeps the already covered vertices).
At each step $i$, pick a hyperedge $C_i \in \EE$, add $C_i$ to $\CC$ and add all the elements of $C_i$ to $W$.
The process is repeated until $W = X$.
In addition, $C_i$ can only be chosen if at least one of its elements has not been previously included in $W$,
i.e.$\!$ $C_i \setminus (\cup_{j=1}^{i-1} C_j) \neq \emptyset$.

How bad can a solution given by this heuristic be (compared to the value of an optimal solution)?
The answer leads to the concept of the \emph{Grundy covering number} of $\HH$ which computes the largest number of steps performed by such a constructive heuristic, or equivalently, the largest number of hyperedges used in the resulting covering \cite{BRESAR2014}.

Let $G=(V,E)$ be a simple graph. For any $v \in V$, let $N(v)$ be the open neighborhood of $v$, i.e.~$N(v) \doteq \{u\in V: (u,v)\in E\}$
(the symbol ``$\doteq$'' will be used recurrently to make definitions) and $N[v]$ be the closed neighborhood of $v$, i.e.~$N[v] \doteq N(v) \cup \{v\}$.

A particular case of the Grundy covering problem occurs when $\HH$ is the hypergraph of the closed neighborhoods of vertices in a graph $G$:
$\HH = (V(G), \EE)$ where $\EE = \{N[v] : v \in V(G)\}$.
Here, the Grundy covering number of $\HH$ is called \emph{Grundy domination number} of $G$ \cite{BRESAR2014}.
Analogously, the \emph{Grundy total domination number} of $G$ is the Grundy covering number of the hypergraph of the open neighborhoods of vertices in $G$ \cite{BRESAR2016}.

In order to illustrate these concepts, consider the graph $G$ of Figure \ref{fig:ex} a) and the
hypergraph of the closed neighborhoods of $V(G)$.
A possible execution of the heuristic would be to pick $N[3]$, then $N[4]$ and finally $N[5]$.
In Figure \ref{fig:ex} b), it is displayed from left to right what happens at each step. A vertex $v$ inside a box means that $N[v]$ is
the chosen set at that step. Filled circles denote those vertices that are already covered (set $W$ of the heuristic).
Note that, at any step, at least one new vertex is covered: in step 2 vertex 5 is covered while in step 3 vertex 1 is covered.
The Grundy domination number of this graph is 3 since it is not possible for the heuristic to perform 4 steps.
\begin{figure}
	\centering
		\includegraphics[scale=0.12]{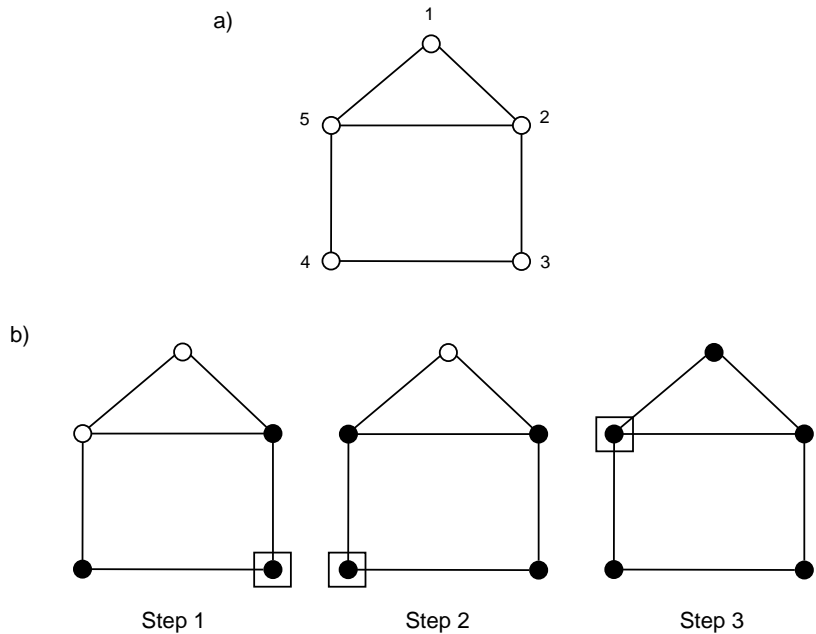}
	\caption{Example: a) graph $G$, b) a maximum legal dominating sequence}  \label{fig:ex}
\end{figure}

The study of these domination parameters was originally motivated by a domination game \cite{P1,P2,P3}, and their associated problems
can model some applications where there are two players with opposing interests.
For instance, consider a city divided into districts where
the municipal government intends to offer a concession per year of a given service (e.g.~Internet connection) to companies.
Each year, a company is located in a district, which is committed to providing its service to both the district and its neighbors.
In return, it requires having at least, within its domain (the district where it is installed and its neighbors), a \emph{captive} district: only that company offers the service to it, for (at least) a year. The goal of the government is to plan which district should be chosen
(for a company to be settled) each year so that, after a while, the city is entirely covered and to maximize the number of companies providing the service (and, thus, to foster long-term competitiveness).
The given problem can be modeled as a Grundy domination problem where each vertex represents a district and two vertices are adjacent if
and only if they represent neighboring districts. Indeed, the Grundy domination number gives the maximum number of companies.
For instance, if the city is modeled as the graph of Figure \ref{fig:ex} a), an optimal schedule is to settle a company on district 3
in year 1, another on district 4 in year 2 and the last one on district 5 in year 3. During the 2nd year, district 5 only receives the
service from the company on district 4. The same happens for district 1 and the company on district 5 during the last year.
In Section \ref{SSCOMPU}, the city of Buenos Aires is considered.

The problems associated to these parameters are \emph{hard} for general graphs.
In \cite{BRESAR2014}, it is proven that the Grundy domination problem is $\NP$-hard for chordal graphs (it is also proven that
this problem is polynomial for trees, cographs or split graphs).
Regarding the total version of this problem, it is $\NP$-hard when $G$ is bipartite \cite{BRESAR2016} (but it is polynomial on trees,
$P_4$-tidy and distance-hereditary bipartites \cite{NASINI2017}).

It is known that one of the most powerful tools to solve $\NP$-hard problems are branch-and-cut algorithms, which are based on Integer Programming. In this work, we introduce a problem that generalizes the Grundy domination and the Grundy total domination, and we propose integer programming formulations for this new problem. In order to obtain good upper and lower bounds, we also design a
heuristic algorithm that combines a greedy strategy with a tabu search. Our approach is exact for instances ranging from 20 to 50 vertices (depending on the edge
density of the graph) and, in particular, by taking only the heuristic algorithm, one can achieve good solutions on large instances.

In Section \ref{SSPROP}, we introduce a general version of the problem and show some useful properties. In particular, we calculate the exact value
of the corresponding parameters for two families of graphs (paths and web graphs). In Sections \ref{SSORIGFORM} and \ref{SSTRENGTH}, we introduce an integer
programming model together with several families of valid inequalities that strengthen its linear relaxation. Besides, we present constraints that remove unnecessary integer points from the solution
space and whose addition results in several formulations. In Section \ref{SSBOUNDS}, we propose a heuristic for obtaining an initial upper
bound and an initial feasible solution of our problem, and a tabu search for improving that initial solution.
In Section \ref{SSCOMPU}, we perform extensive computational experiments to compare the formulations as well as to test two families of valid
inequalities as cuts in a branch-and-cut framework.
We also evaluate the performance of the tabu search. The test instances include randomly generated graphs, structured graphs, classical benchmark instances and two real instances from the aforementioned application in the city of Buenos Aires. Another experiment allows us to formulate a conjecture about the Grundy domination
number on Kneser graphs. Finally, in Section \ref{SSCONCLU} some conclusions are drawn.

Some results contained in this work appeared without proof in the extended abstract \cite{LAGOS2017}.

\subsection{Definitions and notation}

Let $G = (V, E)$ be a simple graph. Also, let $C$ be a subset of vertices of $V$.
Define the function $N\la\_\ra : V \rightarrow \mathcal{P}(V)$, called \emph{neighborhood} of $v$, as follows:
\[ N\la v \ra \doteq \begin{cases}
     N[v], & \textrm{if}~v \in C \\
		 N(v)  & \textrm{if}~v \notin C \end{cases} \]
Here, $\mathcal{P}(V)$ denotes the powerset of $V$.
Assume that no vertex from $V \setminus C$ is isolated in $G$ so that, for all $v\in V$, $N\la v\ra\neq \emptyset$ and $v\in N \la w\ra$ for some $w\in V$.

Some definitions given in \cite{BRESAR2014,BRESAR2016} (which only depends on $G$) are rewritten below in terms of the pair ``$G;C$''.
These definitions abstract the behavior of the heuristic mentioned at the beginning of this work.

A sequence $S=(v_1,\ldots,v_k)$ of distinct vertices is called a \emph{legal sequence} of $G;C$ if
\begin{equation*}
W_i \doteq N\la v_i\ra \setminus \bigcup_{j=1}^{i-1} N\la v_j\ra \ne \emptyset,~~~\forall~i = 2,\ldots,k.
\end{equation*}
By convention, $W_1 \doteq N\la v_1\ra$ which is trivially non empty. For every $i$, the vertices of $W_i$ are said to be
\emph{footprinted} by $v_i$.

For a given sequence $S = (v_1,\ldots,v_k)$, define $\widehat{S} \doteq \{v_1,\ldots,v_k\}$ (i.e.~the set of vertices of the sequence).
Then, $S$ is a \emph{dominating sequence} if $\widehat{S}$ is a dominating set of $G$ (or equivalently, if $\cup_{j=1}^k W_j = V$).

We say that a legal sequence $S$ of $G;C$ is \emph{maximal} if there is no $v \in V \setminus \widehat{S}$ such that $(S,v)$,
i.e.~the sequence $S$ with $v$ appended at the end, is legal.
We say that a legal sequence $S$ of $G;C$ is \emph{maximum} if there is no legal sequence $S'$ of $G;C$ such that $S'$ is longer than $S$.
A maximum legal sequence is also maximal.

Consider again the graph of Figure \ref{fig:ex} and let $C = V$.
Part b) actually shows that the sequence $(3,4,5)$ is legal and dominating:
$W_1 = \{2,3,4\}$ (2, 3 and 4 are footprinted by 3), $W_2 = \{5\}$ (5 is footprinted by 4), $W_3 = \{1\}$ (1 is footprinted by 5)
and $W_1 \cup W_2 \cup W_3 = V$. The sequence is also maximum. On the other hand, the sequence $(5,3)$ is a maximal legal sequence that is not maximum.

\section{The General Grundy Domination Problem} \label{SSPROP}

Let $G = (V, E)$ be a simple graph on $n$ vertices and $C \subseteq V$ such that no vertex from $V \setminus C$ is isolated in $G$.
Let $Hyp(G;C)$ denote the hypergraph $(V, \EE)$ where $\EE = \{N\la v\ra : v \in V\}$.

We define the \emph{general Grundy domination number} of $G;C$, denoted by $\grd(G;C)$, as the Grundy covering number of $Hyp(G;C)$.
It gives rise to the following problem:\\

\medskip

\begin{tabular}{|l|}
\hline
\textsc{General Grundy Domination Problem} (GGDP)\\
{\small \underline{INSTANCE:} a graph $G = (V,E)$ and a set $C \subseteq V$ such that no isolated vertex}\\
~~~~~~~~~~~~~~~~{\small is in $V \setminus C$.}\\
{\small \underline{OBJECTIVE:} obtain $\grd(G;C)$.}\\
\hline
\end{tabular}\\

\medskip

Since the Grundy domination problems mentioned in the introduction are particular cases of our problem, they can be addressed by a tool
that just solves the GGDP: $\grd(G;V)$ is indeed the Grundy domination number of $G$ while the Grundy total domination number is
$\grd(G;\emptyset)$.\\

The following result will be useful to get the general Grundy domination number.
It shows that every optimal solution is a maximum legal sequence and the ``dominating'' property comes for free.
\begin{prop} \label{FIRSTPROP}
Let $G;C$ be an instance of the GGDP and $S$ be a maximal legal sequence of $G;C$.
Then, $S$ is dominating.
Moreover, if $S$ is maximum, then $\grd(G;C)$ is the length of $S$.
\end{prop}
\begin{proof}
If $S$ were not dominating, there would exists a vertex $v$ not footprinted by any element of $S$.
As $G;C$ is an instance of the GGDP, there exists a vertex $w$ such that $v \in N\la w\ra$, and so $w\notin S$.
Thus, the sequence $(S,w)$ is legal since $w$ footprints $v$, which contradicts the maximality of $S$.
Therefore, $S$ is dominating.

The second part of the statement is derived by the fact that legal dominating sequences are (by definition) in one-to-one
correspondence with the solutions provided by the constructive heuristic for $Hyp(G;C)$.
\end{proof}

\subsection{Properties on the GGDP} \label{PROPERTIESGGDP}

As we will see later, knowing a good upper bound will be fundamental for the performance of the exact algorithm that will be proposed
(the worse the bound, the larger the integer linear program that must be solved).
The following result gives a bound on the length of a sequence where the first elements are known.

\begin{prop}
Let $t, k$ be positive integers such that $t \leq k$, $S = (v_1,\ldots,v_t)$ and $S' = (S,v_{t+1},\ldots,v_k)$ be legal sequences.
Then,
$$k \leq n - \biggl| \bigcup_{i=1}^t N\la v_i\ra \biggr| + t.$$
\end{prop}
\begin{proof}
Let $W \doteq \bigcup_{i=1}^t N\la v_i\ra$. If $k=t$, i.e.~$S=S'$, the inequality becomes $|W| \leq n$, which trivially holds.
Now, assume $k>t$. Since $S'$ is legal, vertices $v_{t+1},\ldots,v_k$ must footprint at least one different
vertex from $V \setminus W$, implying that $k - t \leq n - |W|$. Hence, $k \leq n - |W| + t$.
\end{proof}
Using $k = \grd(G;C)$ in the above proposition, we readily get:
\begin{cor} \label{UPPERBOUNDCOR}
Let $t$ be a positive integer such that $t \leq \grd(G;C)$ and $\delta_t$ be the minimum cardinality of $\bigcup_{i=1}^t N\la v_i\ra$ for all legal sequences $(v_1,\ldots,v_t)$.
Then, $\grd(G;C) \leq n - \delta_t + t$.
\end{cor}
Let $m_t \doteq n - \delta_t + t$. For the case $t=1$, finding $m_1$ is linear and provides an upper bound which can be tight sometimes.
For instance, let $G$ be the graph of Figure \ref{fig:ex} and $C = V$.
Here, $\delta_1 = 3$ and, by the corollary, $\grd(G;V) \leq 3 = m_1$.
As the sequence $(3,4,5)$ is legal, $\grd(G;V) = 3$.
The bound $m_1$ was previously derived in \cite{BRESAR2014} for the Grundy domination number (i.e.~$C=V$) and in \cite{BRESAR2016} for the total case (i.e.~$C=\emptyset$),
although their proofs are different from ours, which is valid for any $C$.

Better bounds can be computed by increasing $t$. For instance, let $G$ be the graph of Figure \ref{fig:ub} a) and $C = V$.
Here, $\delta_1 = 2$, $\delta_2 = 4$, $\delta_3 = 6$ and,
by the corollary, we get the upper bounds 7, 6 and 5 respectively.
None of them is tight since $\grd(G;V) = 4$ (the sequence (1,2,4,6) is optimal).
A similar example is given in Figure \ref{fig:ub} b) for the case $C = \emptyset$. Here, $\delta_1 = 1$, $\delta_2 = 3$ and $\delta_3 = 5$,
thus obtaining the upper bounds 7, 6 and 5 respectively, but $\grd(G;\emptyset) = 4$ (the sequence (5,2,6,1) is optimal).
\begin{figure}
	\centering
		\includegraphics[scale=0.12]{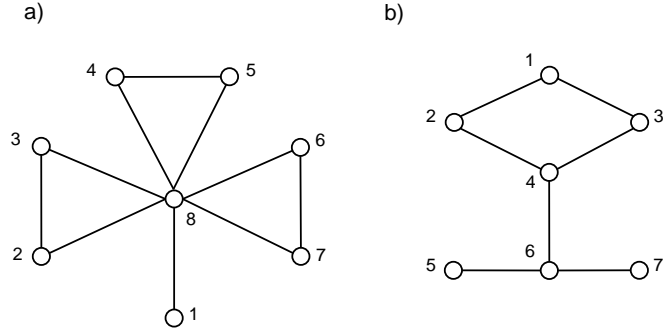}
	\caption{a) case $C=V$, b) case $C=\emptyset$}  \label{fig:ub}
\end{figure}

Another way to shorten the time of optimization is by reducing the size of the input graph.

We say that distinct vertices $u, v$ are \emph{twins} if $N\la u\ra = N\la v\ra$.
If no twin vertices are present in $G;C$, then the instance is called \emph{twin free}.
On the other hand, if $u, v$ are twins in $G;C$, then
$$\grd(G;C) = \grd(G - v; C \setminus \{v\}),$$
where $G - v$ stands for the graph that results from deleting $v$ from $G$.
This rule can be applied recursively until the instance is twin free.

Besides, if $G$ is the (disjoint) union of graphs $G_1$ and $G_2$, then
$$\grd(G;C) = \grd(G_1;C \cap V(G_1)) + \grd(G_2; C \cap V(G_2)).$$
This rule allows us to solve each component separately as the overall time is lower than solving $G;C$ directly.

\subsection{The GGDP on paths and web graphs}

When a new $\NP$-hard optimization problem on graphs is introduced, 
a natural question is what happens on simple structures such as paths and circuits.
Below, we give $\grd(G;C)$ for two families of graphs (paths and web graphs) with any $C$.
Recall that GGDP is already $\NP$-hard on bipartite or chordal graphs.

Let $P_n$ denote an induced path on $n\geq 1$ vertices where $V(P_n)=\{1,\ldots,n\}$.
Note that, for any integer $n\geq 2$ and any $C \subseteq V(P_n)$,
$\delta_1 = 2$ when $\{1,n\} \subseteq C$ and $\delta_1 = 1$ otherwise. Therefore,
\begin{obs} \label{OBSERV1}
Let $n \geq 2$ and $C \subseteq V(P_n)$. If $\{1,n\} \subseteq C$, then $m_1 = n - 1$; otherwise, $m_1 = n$.
\end{obs}

In this context, we introduced in \cite{LAGOS2017} the concept of \emph{good configuration} for a path. Precisely, a subset of vertices
$C$ is a good configuration (\emph{gconf} for short) for $P_n$ if\\
{ \small
\indent \indent (i) $n=1$ and $C = \{1\}$,\\
\indent \indent (ii) $n=2$ and $C \neq \{1,2\}$,\\
\indent \indent (iii) $n\geq 3$ and either\\
\indent \indent \indent (iii.1) $1\notin C$ and $C$ is a gconf for the subpath\\
\indent \indent \indent \indent induced by $\{3,\ldots,n\}$ or\\
\indent \indent \indent (iii.2) $n\notin C$ and $C$ is a gconf for the subpath\\
\indent \indent \indent \indent induced by $\{1,\ldots,n-2\}$.\\ }

For the sake of simplicity, when we say that $C$ is a gconf for a subpath $P'$ of $P$, we are actually referring to the set $C \cap V(P')$. Moreover, in contrast with the definition of GGDP, we allow the subgraph induced by $V(P')\setminus C$ to have isolated vertices in the gconf definition.

From the previous observation, we trivially have
\begin{obs} \label{OBSERV2}
Let $n \geq 1$ and $C \subseteq V(P_n)$. If $C$ is a gconf for $P_n$, then $m_1 = n$.
\end{obs}

The following results give the general Grundy domination number of a path and, at the same time, show where the upper bound $m_1$ is tight.

\begin{prop}
Let $n\geq 1$, $G = P_n$ and $C \subseteq V(G)$.
If $\{1,n\}\subseteq C$ or $C$ is a gconf for $G$, then $\grd(G;C) = m_1$; otherwise, $\grd(G;C) = m_1 - 1$.
\end{prop}
\begin{proof}
If $n=1$, we must have $C=\{1\}$, because $N\la1\ra\neq \emptyset$. So $C$ is a gconf and $\grd(G;C) = m_1 = 1$.
Now, consider the case $n = 2$.
If $C = \{1,2\}$ then $\grd(G;C) = m_1 = 1$.
And, if $C \neq \{1,2\}$ ($C$ is a gconf for $G$) then $\grd(G;C) = m_1 = 2$.

For $n\geq 3$, note that $(1,2,\ldots,n-1)$ is a legal dominating sequence.
Indeed, $1$ footprints $2$ (and itself, if $1\in C$), $2$ footprints $3$ (and $1$, if $1\notin C$), and $i=3,\ldots,n-1$ footprints $i+1$.
Thus, $n-1\leq \grd(G;C)\leq m_1$.
If $\{1,n\}\subseteq C$, by Observation \ref{OBSERV1} we have $\grd(G;C) = m_1 = n-1$.
It remains to consider the case $\{1,n\}\not\subseteq C$. By Observation \ref{OBSERV1}, $m_1 = n$. Then, $\grd(G;C)\in \{m_1-1,m_1\}$,
and it is enough to prove that $C$ is a gconf for $P_n$ if and only if there exists a legal dominating sequence of size $n$.

First, assume that $C$ is a gconf for $P_n$.
We use induction on $n$.
Recall that we have already obtained $\grd(G;C)=2$ for $n=2$, and $\grd(G;C)=1$ is trivial for $n=1$.
For $n\geq 3$, (iii.1) or (iii.2) holds. Assume without loss of generality that (iii.1) holds (the other case is symmetric).
The induction hypothesis ensures the existence of a legal dominating sequence $S$ for $(3,\ldots,n)$, with $|S|=n-2$.
Consider the extended sequence $S' = (1,S,2)$. $S'$ is legal and dominating for $P_n$ with $|S|=n$.

Now, assume that there exists a legal dominating sequence $S$ such that $|S|=n$.
We use again induction on $n$.
$C$ is trivially a gconf for $P_n$ when $n\in\{1,2\}$.
For $n\geq 3$, any dominating legal sequence of length $n$ must start with an endpoint of $P_n$, say $1$, and such a vertex must belong to
$V\setminus C$. Moreover, $2$ must be the last vertex in the sequence, and the $n-2$ vertices between $1$ and $2$ define a legal dominating
sequence for $(3,\ldots,n)$. By the induction hypothesis, $C$ is a gconf for $(3,\ldots,n)$. Since $1\notin C$, $C$ is also a
gconf for $(1,\ldots,n)$.
\end{proof}	

\begin{cor} \label{cor:path}
Let $n\geq 1$, $G = P_n$ and $C \subseteq V(G)$.
If $C$ is a gconf for $G$, then $\grd(G;C)=n$. Otherwise, $\grd(G;C)=n-1$.
\end{cor}

The second (less trivial) family are web graphs \cite{ANNEGRET}.
Let $n, k$ be positive integers such that $n \geq 2(k+1)$.
A \emph{web} $W_n^k = (V,E)$ is a graph with $V = \{0,\ldots,n-1\}$ and $E = \{(i,j) : 0<|i-j|\leq k \text{~or~} |i-j|\geq n-k\}$. See examples in Figure~\ref{fig:webs}.
Note that $N(i) = \{i \ominus k, i \ominus (k-1), \ldots, i \ominus 1, i \oplus 1,\ldots, i \oplus (k-1), i \oplus k\}$, where $\oplus$ and $\ominus$ stand for the addition and subtraction modulo $n$.
Therefore, $m_1=n-2k$ if $C=V$, and $m_1=n-2k+1$ if $C\neq V$.
\begin{figure}
	\centering
		\includegraphics[scale=0.12]{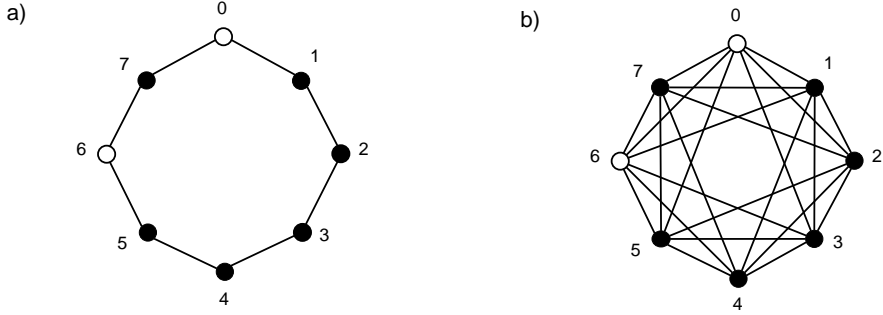}
	\caption{Web graphs: a) $W_8^1$, b) $W_8^3$.}  \label{fig:webs}
\end{figure}

\begin{prop} \label{WebProof}
Let $G$ be the web graph $W_n^k$ and $C \subseteq V(G)$ ($C$ possibly empty).
Then, $\grd(G;C) = m_1$ in the following cases: (i) $C=V$, or (ii) there is $i\in V\setminus C$ such that $V\setminus N[i]$
induces a path $P_t$, $t=n-2k-1$, and $C$ is a gconf for $P_t$.
Otherwise, $\grd(G;C) = m_1-1$.
\end{prop}
\begin{proof}
Consider the sequence $(0, \ldots, n - 2k-1)$ of length $n-2k$.
It is a legal dominating sequence since vertex $0$ footprints $1,\ldots,k$ and $n-k,\ldots,n-1$ (and itself if $0 \in C$),
vertex $1$ footprints $k+1$ (and vertex $0$ if $0 \notin C$) and, if $n>2(k+1)$, vertex $i$ footprints $k+i$ for all $i = 2,\ldots,n-2k-1$. The last footprinted vertex is $n-k-1$. Therefore, we obtain $m_1-1\leq n-2k\leq \grd(G;C)\leq m_1$.
If $C=V$, we are done since $m_1=n-2k$.

Assume now that $C\neq V$. Let $i\in V\setminus C$, and
$$V_i\doteq V\setminus N[i] =\{i\oplus (k+1), i\oplus (k+2),\ldots,i\oplus (n-k-1)\}.$$
Note that $|V_i|=n-2k-1$. It suffices to show that $\grd(G;C) = m_1$ if and only if $V_i$ induces a path and $C$ is a gconf for it. Recall that $m_1=n-2k+1$. The unique way of getting a dominating legal sequence of length $m_1$ is by starting with $i$ and then choosing a vertex that will footprint only one more vertex at each step. This means that, at any step but the last one, we cannot choose a vertex from $N(i)$ because it would footprint $i$ and at least one vertex from $V_i$. So, after $i$, we must choose all the $n-2k-1$ vertices in $V_i$, and finally a vertex from $N(i)$. In addition, $V_i$ must induce a path, otherwise some of its vertices would footprint at least 2 vertices.
Consider the sequence $S' = (i,S,j)$ where $j \in N(i)$ and $S$ is a maximum legal sequence of $G[V_i]$.
In virtue of Corollary~\ref{cor:path}, $|S|=n-2k-1$ if $C$ is a gconf for $G[V_i]$ and $|S|=n-2k-2$ otherwise.
Therefore, $\grd(G;C) = m_1$ if and only if $V_i$ induces a path and $C$ is a gconf for it.
\end{proof}

For example, consider the web graphs depicted in Figure \ref{fig:webs}.
Filled circles denote the vertices of $C$,
i.e.$\!$ $C = \{1,2,3,4,5,7\}$. Upper bounds for $\grd(G;C)$ are $m_1 = 7$ (if $G=W_8^1$) and $m_1 = 3$ (if $G=W_8^3$).
In a), $C$ is neither a gconf for $V\setminus N[0] = \{2,3,4,5,6\}$ nor for $V\setminus N[6] = \{0,1,2,3,4\}$, so a legal sequence of maximum length is $(0,1,2,3,4,5)$.
In b), since $C$ is a gconf for $V\setminus N[0] = \{4\}$, a legal sequence of maximum length is $(0,4,1)$.

\section{Integer programming formulation} \label{SSORIGFORM}

Let $G;C$ be an instance of the GGDP, and $LB$ and $m$ be a lower bound and upper bound on $\grd(G)$, respectively.
We present an ILP formulation for the problem by modeling the iterative process performed by the constructive heuristic
described in the introduction.
For the sake of clarity, through this section, we also present examples based on the graph of Figure \ref{fig:ub} a) with $C = V$ and
$m = 7$.

Legal sequences will be modeled as binary vectors:
a given sequence is represented by an array of $m$ consecutive slots $(s_1,\ldots,s_m)$ where each slot can be empty or store a vertex.
An empty slot is skipped when forming the sequence from its corresponding array.
For instance, the array $(\_,1,\_,2,4,\_,6)$, where ``$\_$'' denotes the empty slot, represents the sequence $(1,2,4,6)$. 
Observe that a sequence $S$ can be represented by ${m\choose m-|S|}$ different arrays.
An empty slot can be interpreted as an innocuous step (without performing any action) in the process, just to attain $m$ steps in total.

Now, for a given array $(s_1,\ldots,s_m)$, let $y_{vi}$ be a binary variable such that $y_{vi} = 1$ if and only if $v$ is in the
$i$th.~slot, i.e.~$s_i = v$.
In order for these variables to reflect an array composed of distinct vertices, the following constraints must hold:
\begin{align}
  & \sum_{v \in V} y_{vi} \leq 1, & \forall~~i = 1,\ldots,m \label{RESTR1}\\
  & \sum_{i=1}^m y_{vi} \leq 1, & \forall~~v \in V          \label{RESTR2}
\end{align}
Constraints \eqref{RESTR1} guarantee that at most one vertex is chosen in each slot.
Constraints \eqref{RESTR2} ensure that each vertex can be chosen no more than once.\\

In order to model the legality of sequences, we consider another set of binary variables that keeps a record of footprinted vertices.
First, let $A_0 = V$, $A_i = A_{i-1}$ if $s_i$ is empty, and
$A_i = A_{i-1} \setminus N\la s_i\ra$ otherwise, for $i = 1,\ldots,m$.
In other words, set $A_i$ contains those vertices available to be footprinted in step $i+1$ and subsequent ones.
In case the slot $s_i$ is empty, no vertex is footprinted at step $i$.
Now, observe that the sequence represented by an array is legal if, for every non-empty $s_i$, $A_{i-1} \cap N\la s_i \ra \neq \emptyset$.

Let $x_{ui}$ be a binary variable such that $x_{ui} = 1$ if and only if $u \in A_i$. In particular, $x_{u0} = 1$ for all $u \in V$ meaning that $x_{u0}$ can be treated as a constant.

The next constraints model the chain of inclusions $A_i \subseteq A_{i-1}$:
\begin{align}
  & x_{ui} \leq x_{u(i-1)}, & \forall~~u \in V,~i = 2,\ldots,m \label{RESTR5}
 \end{align}
(the inclusion $A_1 \subseteq A_0$ naturally holds because $x_{u1} \leq 1$).

Now, we model $A_i \supseteq A_{i-1}$ when $s_i$ is empty and $A_i \supseteq A_{i-1} \setminus N\la s_i\ra$ when $s_i \in V$.
This is equivalent to ask that $u\in A_{i-1} \implies u\in A_i$, if $s_i$ is empty or $s_i \in V \setminus N\la u\ra$.
It can be ensured by the following constraints:
\begin{align}
 & \sum_{v \in N\la u\ra} y_{vi} \geq x_{u(i-1)} - x_{ui}, & \forall~~u \in V,~i = 1,\ldots,m \label{RESTR7}
\end{align}
Indeed, suppose that $s_i$ is empty or $s_i \in V \setminus N\la u\ra$.
Then, the l.h.s.~of \eqref{RESTR7} is zero, implying $x_{ui} \geq x_{u(i-1)}$, as desired.
On the other hand, if $s_i \in N\la u\ra$, then \eqref{RESTR7} becomes redundant.

In order to model $A_i \subseteq A_{i-1} \setminus N\la s_i\ra$ when $s_i$ is not empty, none of $u \in A_{i-1}$ such that
$s_i \in N\la u\ra$ can belong to $A_i$.
Constraints \eqref{RESTR4}, presented below, remove those elements from $A_i$:
\begin{align}
 & x_{ui} + \sum_{v \in N\la u\ra} y_{vi} \leq 1, & \forall~~u \in V,~ i = 1,\ldots,m \label{RESTR4}
\end{align}
Therefore, constraints~\eqref{RESTR5}-\eqref{RESTR4} guarantee the correct valuation of the $x$ variables.

Finally, in order for the sequence (represented by the array) to be legal, for any $i\geq 2$ such that $s_i$ is not empty, there must
be at least one available vertex of $N\la s_i\ra$ at step $i-1$, which becomes footprinted at $i$:
\begin{align}
 & y_{vi} \leq \sum_{u \in N\la v\ra} (x_{u(i-1)} - x_{ui}), & \forall~~v \in V,~ i = 2,\ldots,m \label{RESTR3}
\end{align}
Indeed, if $y_{vi} = 1$ (i.e.~$s_i=v$), then the r.h.s.~of \eqref{RESTR3} is positive implying that some $u \in N\la v\ra$ must satisfy $x_{u(i-1)} > x_{ui}$,
that is $x_{u(i-1)} = 1$ and $x_{ui} = 0$.\\

The previous families of constraints are enough to model legal sequences but not all of them are necessary.
In fact, we can dispense with constraints \eqref{RESTR7} (as proved below).
It leads to the following formulation, which we call $F_1$, that obtains the array that maximizes the amount of non-empty slots:

\begin{align} \setcounter{equation}{0}
 & \max \sum_{i=1}^m \sum_{v \in V} y_{vi} & \notag \\
 & \textrm{subject to} & \notag \\
 & ~~~~\sum_{v \in V} y_{vi} \leq 1, & \forall~~i = 1,\ldots,m \\
 & ~~~~\sum_{i=1}^m y_{vi} \leq 1, & \forall~~v \in V          \\
 & ~~~~x_{ui} \leq x_{u(i-1)}, & \forall~~u \in V,~i = 2,\ldots,m \\ \setcounter{equation}{4}
 & ~~~~x_{ui} + \sum_{v \in N\la u\ra} y_{vi} \leq 1, & \forall~~u \in V,~ i = 1,\ldots,m \\
 & ~~~~y_{vi} \leq \sum_{u \in N\la v\ra} (x_{u(i-1)} - x_{ui}), & \forall~~v \in V,~ i = 2,\ldots,m \\
 & ~~~~x_{vi}, y_{vi} \in \{0,1\} & \forall~~v \in V,~ i = 1,\ldots,m \notag
\end{align}

\begin{thm} \label{thm:F1correct}
Formulation $F_1$ gives the value of $\grd(G;C)$.
\end{thm}
\begin{proof}
Let $(s_1,\ldots,s_m)$ be the array defined by an optimal solution $(x^*,y^*)$ of $F_1$.
Precisely, for each $i = 1,\ldots,m$,
$$s_i= \left\{ 
\begin{array}{ll}
\_, & \text{if } \sum_{v \in V} y^*_{vi}=0,\\
v, & \text{if } \exists~v \in V~\text{such that}~y^*_{vi}=1.
\end{array}
\right.$$
Note that $s_i$ is uniquely defined by \eqref{RESTR1}.
Consider the sequence $S = (v_1,\ldots,v_k)$ represented by that array. We have to show that $S$ is legal and has maximum length.
Since the non-empty slots have distinct vertices by \eqref{RESTR2}, so do $S$. 
Now, consider $v_i\in S$ for some $i\geq 2$. There is some $j\geq i$ such that $s_j = v_i$ (i.e.~$y^*_{v_i j}=1$) in the array.
By constraints \eqref{RESTR5} and \eqref{RESTR3}, there exists $u\in N\la s_j\ra$ such that $x^*_{uj-1}=1$.
Besides, by constraints \eqref{RESTR5} and \eqref{RESTR4}, we have $x^*_{ul}=1$ and $\sum_{v\in N\la u\ra} y^*_{vl}=0$, for all
$l=1,\ldots,j-1$. It follows that $s_l\notin N\la u\ra$ for all $l=1,\ldots,j-1$. Therefore,
$$u\in N\la s_j\ra \setminus \bigcup_{\substack{l=1 \\ s_l \in V}}^{j-1} N\la s_l\ra = N\la v_i\ra \setminus \bigcup_{l=1}^{i-1} N\la v_l\ra,$$
which implies that $S$ is legal.

Now, suppose that $S'$ is a legal sequence longer than $S$, e.g.~$S'=(v'_1,\ldots,v'_{k+1})$. As $m \geq \grd(G;C) \geq k+1$,
$S'$ fits in an array. Consider $(v'_1, \ldots, v'_{k+1}, \_, \ldots, \_)$ and the sets $A'_0 = V$,
$A'_i = A'_{i-1} \setminus N\la v'_i\ra$ for all $i = 1,\ldots,k+1$ and $A'_i = A'_{k+1}$ for all $i = k+2,\ldots,m$ (if $m > k+1$).
Clearly, the array and the sets correspond to a feasible integer point $(x,y)$ with objective value of $k+1$, which leads to a contradiction.
\end{proof}

\section{Strengthening the formulation} \label{SSTRENGTH}

In this section, we explore several ways of strengthening the linear relaxation of $F_1$ by the addition of inequalities.

\subsection{Symmetry-breaking inequalities} \label{FORMLIST}

In an ILP model, two integer solutions are called \emph{symmetric} if they share the same objective value and one
can be obtained from the other by swapping its components \cite{MARGOT}. In a broader sense, symmetric solutions have a common characteristic
(for instance, both represent the same solution of the problem) so that one can be eliminated without losing the correctness of the model.
Sometimes, ILP models with several symmetric solutions may have poor performance and one way to reduce the number of symmetric
solutions is to incorporate symmetric-breaking constraints to the formulation.
Certainly, as we have pointed out before, it is not necessary to consider constraints \eqref{RESTR7} in the formulation but when added they cut off some symmetric solutions that otherwise are present in $F_1$.
Below, we present two families of constraints that also do this job.
\begin{itemize}
\item For a given legal sequence $(v_1,\ldots,v_k)$, consider an array such that the first elements are those from the sequence
and the remaining slots are empty. That is, an array $(s_1,\ldots,s_m)$ such that $s_i = v_i$ for $i \in \{1,\ldots,k\}$ and
$s_i = \_$ for $i \in \{k+1,\ldots,m\}$. 
We call it \emph{canonical}.
For instance, the sequence (1,2,4,6) is represented by the canonical array $(1,2,4,6,\_,\_,\_)$.
Those solutions associated to non-canonical arrays can be removed by
\begin{align}
 & \sum_{v \in V} y_{vi} \leq \sum_{v \in V} y_{v(i-1)}, & \forall~~i = 2,\ldots,m \label{RESTR8NEW}
\end{align}
which means that, if for some $i \geq 2$, $s_i$ is not empty, then $s_{i-1}$ is not empty too.
\item As optimal sequences are dominating, we can impose that every integer solution represents a dominating sequence:
\begin{align}
 & \sum_{i=1}^m \sum_{v \in N\la u\ra} y_{vi} \geq 1, & \forall~~u \in V \label{RESTR9NEW}
\end{align}
\end{itemize}
Different combinations of the symmetry-breaking constraints presented above give rise to 8 formulations listed in the following table.
For instance, $F_2$ is the same as $F_1$ plus constraints \eqref{RESTR7}:
\begin{center}
\begin{tabular}{|c|c|c|c|c|c|c|c|c|}
\hline
 Form. & \eqref{RESTR1} & \eqref{RESTR2} & \eqref{RESTR5} & \eqref{RESTR7} & \eqref{RESTR4} & \eqref{RESTR3} & \eqref{RESTR8NEW} & \eqref{RESTR9NEW} \\
\hline
 $F_1$ & \tick          & \tick          & \tick          &                & \tick          & \tick          &                               & \\
 $F_2$ & \tick          & \tick          & \tick          & \tick          & \tick          & \tick          &                               & \\
 $F_3$ & (only $i=1$)   & \tick          & \tick          &                & \tick          & \tick          & \tick                         & \\
 $F_4$ & (only $i=1$)   & \tick          & \tick          & \tick          & \tick          & \tick          & \tick                         & \\
 $F_5$ & \tick          & \tick          & \tick          &                & \tick          & \tick          &                               & \tick \\
 $F_6$ & \tick          & \tick          & \tick          & \tick          & \tick          & \tick          &                               & \tick \\
 $F_7$ & (only $i=1$)   & \tick          & \tick          &                & \tick          & \tick          & \tick                         & \tick \\
 $F_8$ & (only $i=1$)   & \tick          & \tick          & \tick          & \tick          & \tick          & \tick                         & \tick \\
\hline
\end{tabular}
\end{center}
In some rows, constraints \eqref{RESTR1} are omitted (except for the case $i=1$) since they are dominated by \eqref{RESTR8NEW}.

All of the enumerated formulations are correct (in the sense that they deliver $\grd(G;C)$) since $F_1$ is correct and the other
ones just delete unnecessary solutions.
In particular, two of them deserve special attention:
\begin{obs}
There is a one-to-one correspondence between legal sequences of $G;C$ and integer solutions of $F_4$.
\end{obs}
\begin{obs}
There is a one-to-one correspondence between dominating legal sequences of $G;C$ and integer solutions of $F_8$.
\end{obs}
These observations come from the following facts. Each legal sequence $(v_1,\ldots,v_k)$ is represented by a unique canonical array $(s_1,\ldots,s_m)$, defined by $s_i=v_i$ for $i=1,\ldots,k$, and $s_i=\_$ for $i=k+1,\ldots,n$. Besides, the vertices footprinted at iteration $i$ are precisely $A_{i-1}\setminus A_{i}$, for $i=1,\ldots,m$ (recall the definition of $A_i$ in Section~\ref{SSORIGFORM}). So, let $(x,y)$ be a feasible solution of $F_4$. Since it is feasible to $F_1$, $y$ defines an array $(s_1,\ldots,s_m)$ that gives a legal sequence $S$, as in the first part of the proof of Theorem~\ref{thm:F1correct}. Due to \eqref{RESTR8NEW}, such an array is exactly the canonical array representing $S$. In addition, as shown in Section~\ref{SSORIGFORM}, inequalities \eqref{RESTR5}, \eqref{RESTR7} and \eqref{RESTR4} precisely define $A_i$ as $\{u\in V:x_{ui}=1\}$. This means that every legal sequence and the vertices footprinted at each iteration are mapped to a single feasible point $(x,y)$ in $F_4$. Besides, the addition of constraint~\eqref{RESTR9NEW} in $F_8$ restricts this mapping to dominating sequences.\\

Although one would expect that the fewer the number of solutions is, the smaller the size of the branch-and-bound tree will be,
the addition of symmetry-breaking inequalities not always help to improve the optimization.
Therefore, we carried out computational experiments in order to determine which formulation performs better.
Results are reported in Section \ref{SSCOMPU}.

\subsection{Removing non-optimal solutions} \label{NONOPTIMALSECTION}

If a lower bound $LB$ is known, solutions corresponding to sequences of size less than $LB$ are not optimal and, thus,
are not needed. They can be removed by forcing $\sum_{i=1}^m \sum_{v \in V} y_{vi}$ to be at least $LB$.
Even better, we can impose that the first $LB$ slots must be used:
\begin{align}
 & \sum_{v \in V} y_{vi} = 1, & \forall~~i = 1,\ldots,{LB} \label{NONOPTIMALREMOVE}
\end{align}
These constraints can be incorporated to any of the 8 formulations. In that case, some constraints may become dominated and can be omitted:
\eqref{RESTR1} for $i=1,\ldots,LB$ and \eqref{RESTR8NEW} for $i=2,\ldots,LB$.

\subsection{Valid inequalities} \label{VALIDINEQSECTION}

It is known that a cutting-plane algorithm is one of the most efficient tools to deal with an integer linear programming problem \cite{WOLSEY}.
A cutting-plane algorithm tries to strengthen the linear relaxation by adding violated valid inequalities, known as \emph{cuts}.
There are two types of them: general cuts (such as \emph{Gomory cuts}) that do not take advantage of the problem structure, or specific cuts
that exploit the properties of the problem.
In this section, we present some inequalities that are valid for the formulations given above, which will be the basis for a cutting-plane
algorithm.

Let $P_i$ be the convex hull of the set of integer feasible solutions in $F_i$.

Given a non-empty set of vertices $U$ and a positive integer $r \leq |U|$, let $N^r\la U\ra$ denote the subset of vertices with exactly $r$ neighbors in $U$, i.e.
$$N^r\la U\ra \doteq \{v \in V : |N\la v\ra \cap U| = r\}.$$
Observe that $V$ can be partitioned into sets $N^r\la U\ra$ for $r = 0,\ldots,|U|$.
For example, consider the graph of Figure \ref{fig:vi} where $C = V$. Filled circles represent vertices of $U$, i.e.~$U = \{2,8\}$.
Then, $V = N^0\la U\ra \cup N^1\la U\ra \cup N^2\la U\ra$, where $N^0\la U\ra = \{7\}$, $N^1\la U\ra = \{1,4,5,6\}$ (vertices inside a
rhombus \rlap{\BigDiamondshape}\SmallCircle) and $N^2\la U\ra = \{2,3,8\}$ (vertices inside a square \rlap{\BigSquare}\SmallCircle).
\begin{figure}
	\centering
		\includegraphics[scale=0.12]{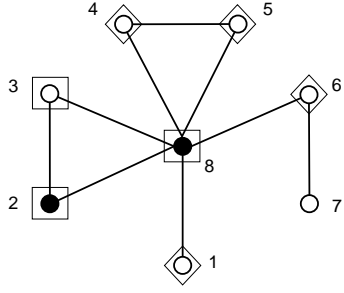}
	\caption{An example with $U = \{2,8\}$}  \label{fig:vi}
\end{figure}

Now, we present a general family of valid inequalities. Consider:
\begin{itemize}\itemsep=0cm
\item an integer $i \in \{2,\ldots,m\}$,
\item an integer $k \in \{1,\ldots,i\}$,
\item a non-empty set $U$ of $p$ vertices,
\item a set $N$ (possibly empty) such that $N \subseteq N^p\la U\ra$,
\item a non-empty set $W$ of $t$ vertices, $W \doteq \{w_1, \ldots, w_t\}$, such that 
\begin{enumerate}
\item[H1)] $W \subseteq N^p\la U\ra \setminus N$,
\item[H2)] $N\la w_{r+1}\ra \subseteq N\la w_r\ra$ for all $r = 1,\ldots,t-1$, and
\item[H3)] $N\la v\ra \subseteq N\la w_t\ra$ for all $v \in N$,
\end{enumerate}
\item integers $j_1, \ldots, j_{t+1} \in \{1,\ldots,i\}$ such that $j_1 = 1$, $j_{t+1} = i$ and $j_r \leq j_{r+1}$ for all $r = 1,\ldots,t$.
\end{itemize}
We define the $(i,k,U,N,W,j_1,\ldots,j_{t+1})$-inequality as:
\begin{equation} \label{SUPERNOVA}
\sum_{u \in U} x_{ui} + \sum_{v \in N} y_{vi} + \sum_{r=1}^t \sum_{j=j_r}^{j_{r+1}} y_{w_r j} + \sum_{v \in N^p\la U\ra} (p-1) y_{vk} + \sum_{q=1}^{p-1} \sum_{v \in N^q\la U\ra} q y_{vk} \leq p
\end{equation}

Before formally proving its validity, we present an example that can make it easier to understand.
Consider again the instance of Figure \ref{fig:vi} with $C = V$ and $m = 7$.
Let $i = 5$, $k = 4$, $U = \{2,8\}$ (thus, $p = 2$), $N = \{2\}$, $W = \{w_1 = 8, w_2 = 3\}$ (thus, $t = 2$), $j_1 = 1$, $j_2 = 3$ and
$j_3 = 5$.
As $N$ and $W$ are disjoint subsets of $N^p \la U\ra$ and $N\la 2\ra \subseteq N\la 3\ra \subseteq N\la 8\ra$,
the $(5,4,\{2,8\},\{2\},\{8,3\},1,3,5)$-inequality, i.e.
\begin{multline*}
\overbrace{x_{25} + x_{85}}^{\sum_{u \in U} x_{ui}} + \overbrace{y_{25}}^{\sum_{v \in N} y_{vi}}
  + \overbrace{y_{81} + y_{82} + y_{83} + y_{33} + y_{34} + y_{35}}^{\sum_{r=1}^t \sum_{j=j_r}^{j_{r+1}} y_{w_r j}} \\
  + \underbrace{y_{24} + y_{34} + y_{84}}_{\sum_{v \in N^p\la U\ra} (p-1) y_{vk}}
	+ \underbrace{y_{14} + y_{44} + y_{54} + y_{64}}_{\sum_{q=1}^{p-1} \sum_{v \in N^q\la U\ra} q y_{vk}} \leq 2
\end{multline*}
is well defined.
Below, there is a representation of variables ``$y$'' from the inequality as a matrix having a row per vertex and a column per step.
\begin{center}
\begin{tabular}{|c|c|c|c|c|c|c|c|}
\hline
       & Step 1 & Step 2 & Step 3 & Step 4  & Step 5  & Step 6 & Step 7 \\
Vertex &        &        &        & $k = 4$ & $i = 5$ &        &  \\
       & $j_1 = 1$ &    & $j_2 = 3$ &      & $j_3 = i$ &       &  \\
\hline
  8    & $y_{81}$ & $y_{82}$ & $y_{83}$ &   $y_{84}$ &          & & \\
  3    &          &          & $y_{33}$ & $2 y_{34}$ & $y_{35}$ & & \\
  2    &          &          &          &   $y_{24}$ & $y_{25}$ & & \\
  1    &          &          &          &   $y_{14}$ &          & & \\
  4    &          &          &          &   $y_{44}$ &          & & \\
  5    &          &          &          &   $y_{54}$ &          & & \\
  6    &          &          &          &   $y_{64}$ &          & & \\
  7    &          &          &          &            &          & & \\
\hline
\end{tabular}
\end{center}
Variables ``$y$'' associated to vertices from $N \cup W$ form a \emph{ladder} ($y_{81}$, $y_{82}$, $y_{83}$, $y_{33}$, $y_{34}$, $y_{35}$
and $y_{25}$). Due to $N\la 2\ra \subseteq N\la 3\ra \subseteq N\la 8\ra$, only one of these variables can be set to 1, so they contribute at most
one unit to the l.h.s. Also, there are variables ``$y$'' in the 4th \emph{column} ($y_{14}$, $\ldots$, $y_{84}$ except $y_{74}$).
They can contribute at most one unit since any feasible solution satisfies constraints \eqref{RESTR1}.
Note that the coefficient of $y_{34}$ is 2 as this variable belongs to both the ladder and the column.

If one variable ``$y$'' from the ladder is set to 1, it means that some vertex from $\{8,3,2\}$ is chosen at some step less than or equal to 5.
Then, vertices 2 and 8 are footprinted so that $x_{25} = x_{85} = 0$. It implies that the l.h.s.~is at most 2 (one unit from the ladder
and, possibly, one unit from the 4th column) and therefore the inequality is valid.
If none of the variables from the ladder is set to 1 but some ``$y$'' from the column is, vertex 8 is footprinted at step 4 and
then $x_{85} = 0$. Hence, the l.h.s.~of the inequality is again at most 2 (one unit from the column and, possibly, one from $x_{25}$).
If no variable ``$y$'' from the inequality is set to 1, then it reduces to $x_{25} + x_{85} \leq 2$ which is trivially valid.

Now that we know it is valid, we wonder if it is not redundant for a given formulation $F_i$.
We can check it with a linear program, say $LP_i$, that maximizes the l.h.s.~of the inequality subject to the linear constraints of $F_i$.
Let $z^*_i$ denote the objective value of $LP_i$. Such value must be greater than 2 for the inequality to cut the linear relaxation of
$F_i$.
For $i \in \{1,2,5,6\}$, we obtain $z^*_i = 3.89$. For $i \in \{3,4,7,8\}$, we obtain $z^*_i = 3.80$.

\begin{thm} \label{VERYGENERALINEQ}
The $(i,k,U,N,W,j_1,\ldots,j_{t+1})$-inequality is valid for $P_1$.
\end{thm}
\begin{proof}
Let $(x,y)$ be a feasible integer point of $P_1$.
Define
$\Sigma_N \doteq \sum_{v \in N} y_{vi}$,
$\Sigma_W \doteq \sum_{r=1}^t \sum_{j=j_r}^{j_{r+1}}$ $y_{w_r j}$,
$\Sigma_p \doteq \sum_{v \in N^p\la U\ra} (p-1) y_{vk}$,
$\Sigma_q \doteq \sum_{q=1}^{p-1} \sum_{v \in N^q\la U\ra} q y_{vk}$ and
$s \doteq p - \sum_{u \in U} x_{ui}$ ($s$ is the number of vertices $u \in U$ such that $x_{ui} = 0$).
We prove $\Sigma_N + \Sigma_W + \Sigma_p + \Sigma_q \leq s$.

First, note that, if a vertex $v \in N^q\la U \ra$ is chosen in a step from $1,\ldots,i$, then $v$ footprints $q$ vertices from $U$
implying that $1 \leq q \leq s$. Since no more than one vertex can be chosen in step $k$ and $N^1\la U\ra$, $N^2\la U\ra$, $\ldots$,
$N^p\la U\ra$ are disjoint sets of vertices, $\Sigma_q\leq s$ and $\Sigma_p + \Sigma_q \leq p-1$.

Suppose that $s < p$. Then, some vertex from $U$ is not footprinted in step $i$.
Therefore, no vertex from $N^p \la U\ra$ can be chosen at steps $1,\ldots,i$, implying that $\Sigma_N = \Sigma_W = \Sigma_p = 0$.
As $\Sigma_q \leq s$, the inequality is valid.

Now, suppose that $s = p$. Then, all vertices from $U$ are footprinted along steps $1,\ldots,i$. If some vertex from $N$ is chosen at step $i$,
i.e.~$\Sigma_N = 1$, then it is not possible to choose any vertex from $W$ in a step previous to $i$ due to hypotheses H2-H3,
implying that $\Sigma_W = 0$. On the other hand, if $\Sigma_N = 0$ and $y_{w_r j} = 1$ for some $r\in \{1,\ldots,t\}$ and
$j \in \{j_r, \ldots, j_{r+1}\}$, H2 guarantees that it is not possible to choose vertices from $w_1,\ldots,w_{r-1}$ in steps $1,\ldots,j$
nor vertices from $w_{r+1},\ldots,w_t$ in steps $j,\ldots,i$. Therefore, $\Sigma_N + \Sigma_W \leq 1$.
Since $\Sigma_p + \Sigma_q \leq p-1 = s-1$, validity follows.
\end{proof}
These inequalities are, sometimes, stronger versions of constraints from $F_1$.
For instance, constraint \eqref{RESTR2} for any $v$ is dominated by the ($m$,$m$,$\{u\}$,$\emptyset$, $\{v\}$,$1$,$m$)-inequality for any $u \in N\la v\ra$,
which is $x_{um}+\sum_{j=1}^m y_{vj} \leq 1$.

Let $u \in V$ and $i = 2,\ldots,m$. Another example arises when there exists a vertex $w\in N\la u\ra$
such that $N\la x\ra \subseteq N\la w\ra$ for every $x\in N\la u\ra \setminus \{w\}$.
The $(i,i,\{u\},N\la u\ra \setminus \{w\},\{w\},1,i)$-inequality, i.e.~$x_{ui} + \sum_{v\in N\la u\ra} y_{vi} + \sum_{j=1}^{i-1} y_{wj} \leq 1$
dominates the constraint \eqref{RESTR4} corresponding to that $u$ and $i$.

The following two subfamilies of $(i,k,U,N,W,j_1,\ldots,j_{t+1})$-inequalities deserve special attention. They can be efficiently separated and have shown to be effective to strengthen the relaxation in our computational experiments (see Section~\ref{SSCOMPU}).
\begin{itemize}
\item \emph{Type I}: Let $i \in \{2,\ldots,m\}$, $u \in V$ and $w \in N\la u\ra$.
The $(i,i,\{u\},\emptyset,\{w\},1,i)$-inequality is
$$x_{ui} + \sum_{j=1}^i y_{w j} \leq 1.$$
\item \emph{Type II}: Let $i \in \{2,\ldots,m\}$, $k \in \{1,\ldots,i\}$, $u_1, u_2 \in V$ such that $u_1 \neq u_2$,
and $w \in N^{\cap} \doteq N\la u_1\ra \cap N\la u_2\ra$. The $(i,k,\{u_1,u_2\},\emptyset,\{w\},1,i)$-inequality is
$$x_{u_1 i} + x_{u_2 i} + \sum_{j=1}^i y_{wj} + \sum_{v \in N^{\cup}} y_{vk} \leq 2,$$
where $N^{\cup} \doteq N\la u_1\ra \cup N\la u_2\ra$.
Note that the coefficient of $y_{wk}$ is 2 in the l.h.s.
\end{itemize}

\section{Obtaining lower and upper bounds} \label{SSBOUNDS}

In this section, we devise an algorithm that provides an initial lower and upper bound of $\grd(G;C)$, which
we call \textsc{GetInitialBounds}, and a metaheuristic based on a tabu search that tries to improve the length of the legal sequence
associated with the initial lower bound.
Both algorithms make heavy use of the following greedy heuristic:

\begin{algorithm}[H] \small
  \SetKwInOut{Input}{Input}
  \SetKwInOut{Output}{Output}
  \Input{A legal sequence $S = (v_1, \ldots, v_k)$ of $G;C$.}
  \Output{A maximal legal sequence $(v_1, \ldots, v_{k'})$ with $k' \geq k$.}
  $W \leftarrow \bigcup_{j=1}^k N\la v_j\ra$\;
  \While{$W \neq V$}{
    $Cand \leftarrow \{ v \in V \setminus S : N\la v\ra \setminus W \neq \emptyset \}$\;
		Choose $v \in Cand$ such that $|N\la v\ra \setminus W| + \beta$ is minimum\;
		$W \leftarrow W \cup N\la v\ra$\;
		Append $v$ to $S$\;
  }
	Return $S$ and exit\;
 \caption{\textsc{Maximalize}}
\end{algorithm}

\medskip

Let $f(v) \doteq |N\la v\ra \setminus W|$, i.e.~the amount of neighbors that can be footprinted by $v$ at the current step.
The heuristic basically chooses those vertices that minimize $f(v)$.
Above, $\beta$ is a function returning a random real number from $[0,1)$ distributed uniformly (each call generates a different
number). Adding $\beta$ to $f(v)$ acts as a tie breaker mechanism: for two different vertices with the same image under $f$, the choice of one
of them will be decided randomly.

\subsection{Initial bounds}

The procedure \textsc{GetInitialBounds} basically explores every legal sequence of 3 vertices and finds $m_3$ (see Corollary \ref{UPPERBOUNDCOR})
together with a maximal legal sequence.
From now on, and for the sake of clarity, we assume that $\grd(G;C) \geq 3$ as the pseudocode of \textsc{GetInitialBounds} can be
easily modified to handle instances such that $\grd(G;C) \leq 2$.

For any $v_1, v_2$, let $v_1 \antes v_2$ denote the expression $N\la v_2\ra \setminus N\la v_1\ra \neq \emptyset$ which
means that the sequence $(v_1, v_2)$ is legal. Note that the relation $\antes$ can be stored in memory as a boolean matrix for fast access. The algorithm is displayed below:

\medskip

\begin{algorithm}[H] \small
  \SetKwInOut{Input}{Input}
  \SetKwInOut{Output}{Output}
  \Input{An instance $G;C$.}
  \Output{A maximal legal sequence $S$ of $G;C$ and the upper bound $m_3$.}
  $\delta_3 \leftarrow n$\;
	$S \leftarrow \{v_1\}$\;
	\For{every $v_1, v_2 \in V$ such that $v_1 \antes v_2$}{
    \For{every $v_3 \in V$ such that $v_1 \antes v_3$ and $v_2 \antes v_3$}{
      $W \leftarrow N\la v_1\ra \cup N\la v_2\ra$\;
      \If{$N\la v_3\ra \setminus W \neq \emptyset$}{
        $W \leftarrow W \cup N\la v_3\ra$\;
        \If{$|W| < \delta_3$}{
          $\delta_3 \leftarrow |W|$\;
          $T \leftarrow \textsc{Maximalize}(v_1,v_2,v_3)$\;
          \If{the length of $T$ is greater than $S$}{
					  $S \leftarrow T$\;
				  }
				}
			}
		}
	}
	Return $m_3 = n - \delta_3 + 3$ and $S$, and exit\;
 \caption{\textsc{GetInitialBounds}}
\end{algorithm}

\medskip

Note that the previous algorithm looks for a legal sequence by maximalizing $(v_1,v_2,v_3)$ each time $\delta_3$ is updated and keeping the
longest one. We also experimented by maximalizing \emph{every} legal sequence of size 3 and, although this procedure yields a little better lower bound, it consumes a long CPU time. In fact, the tabu search reaches the same lower bound in less time.

\subsection{A tabu search based heuristic}

Now, we present the tabu search algorithm. Basically, for a given $k$, it explores (not necessarily legal) sequences of size $k$ 
and finishes its execution when:
\begin{itemize}
\item a maximal legal sequence of size (at least) $k$ is found,
\item a limit in elapsed time or number of iterations is reached (it \emph{fails} for short).
\end{itemize}
Then, the initial sequence $S$ given by \textsc{GetInitialBounds} can be improved by repeatedly invoking the tabu search with
$k$ equal to the length of $S$ plus one.\\

We recall that \emph{tabu search} is a metaheuristic method where a local search algorithm is equipped with an additional mechanism
that prevents from getting stuck in local optima \cite{GLOVER}. Next, we define these concepts briefly.

Consider a minimization problem where $\mathscr{S}$ is its solution space and $f : \mathscr{S} \rightarrow \mathbb{R}_+$ is the
objective function. 
For each solution $s$, consider a \emph{neighborhood} $N(s) \subseteq \mathscr{S}$ satisfying the following properties:
\begin{itemize}
\item if $s, s' \in \mathscr{S}$ are neighbors, it is inexpensive (from the computational point of view) to
compute $s'$ from $s$, and $f(s')$ from $f(s)$,
\item for any $s, s' \in \mathscr{S}$, there exist solutions $s = s_1, s_2, \ldots, s_r = s'$ such that $s_i$ and $s_{i+1}$ are neighbors for all $i = 1,\ldots,r-1$.
\end{itemize}
Let $\mathscr{P}$ be a set of \emph{features} and $R \subseteq \mathscr{S} \times \mathscr{P}$ such that
$(s, p) \in R$ if $s$ presents the feature $p$.
In general, neighboring solutions share most of their features.

In a tabu search, a sequence of solutions $s_1, s_2, \ldots$ is generated from an initial solution $s_0 \in \mathscr{S}$
and \emph{movements} from a solution $s_i$ to another $s_{i+1} \in \arg\min_{s \in N'(s_i)} f(s)$, where $N'(s_i)$ is a subset of
$N(s_i)$ described below.
In the $i$th iteration, some feature of $s_{i+1}$ is stored in a set $L$ called \emph{tabu list}. This list is used to store whether a
movement is allowed or forbidden, and $N'(s)$ has those movements from $N(s)$ that are allowed, i.e.
\[ N'(s) = \{ s' \in N(s) : (s',p) \notin R ~~ \forall~p \in L \}. \]
When a new feature $p$ is inserted into $L$, a value associated to that feature is set to a non negative integer referred to as
\emph{tabu tenure}. The value is called \emph{time of live} of the feature $p$, denoted by $live(p)$, and represents the number of
remaining iterations in which $p$ still belongs to the list $L$. The time of live is decreased by one unit in each iteration and, when
it reaches zero, $p$ is removed from $L$.

In our case, these concepts are instantiated as follows:
\begin{itemize}
\item \emph{Search space and objective function}. $\mathscr{S}$ contains every sequence of $k$ different vertices.
For any $S = (v_1, \ldots, v_k) \in \mathscr{S}$, define the \emph{set of conflicting indexes} as
\[ \mathcal{I}(S) \doteq \{ i : N\la v_i\ra \setminus \bigcup_{j=1}^{i-1} N\la v_j\ra = \emptyset \}. \]
Any vertex $v_i$ such that $i \in \mathcal{I}(S)$ is called \emph{conflicting}.
The objective function is defined as $f(S) = |\mathcal{I}(S)|$.
\item \emph{Stopping criterion}. The algorithm stops when $S$ is a legal sequence, i.e.~$f(S) = 0$, or an iteration/time
limit is reached, meaning that the algorithm fails.
In the first case, $S$ is maximalized and returned.
\item \emph{Initial solution}. It is generated by randomly picking $k$ different vertices with a uniform distribution.
\item \emph{Set of features and tabu list}. Here, vertices are features. The tabu list $L$ contains those vertices marked as tabu.
\item \emph{Neighborhood of a solution}. From a sequence $S = (v_1, \ldots, v_k)$, a neighbor $S' = (v'_1, \ldots, v'_k)\in N'(S)$
is constructed through two movements, within the restricted neighborhoods $N'_1(S)$ and $N'_2(S)$ respectively, which define
$N'(S) = N'_1(S) \cup N'_2(S)$.
In both movements, two vertices are involved and, for the movement to be allowed, one of them should not be marked as tabu.
First, define the following set
\[ \mathcal{J}(S) \doteq \{ j : N\la v_i\ra \cap N\la v_j\ra \neq \emptyset,~ \text{for some}~i \in \mathcal{I}(S),~ i>j \}. \]
If $j \in \mathcal{J}(S)$ due to some $i \in \mathcal{I}(S)$ then $v_j$ \emph{conflicts} with $v_i$
in the sense that $v_j$ is a vertex that occurs before $v_i$ in $S$ and both share elements from their neighborhoods.
Now, we can define the two movements:
\begin{itemize}
\item $N'_1$: \emph{Swap a vertex from the sequence with one from outside}.
Let $l \in \mathcal{I}(S) \cup \mathcal{J}(S)$ and $z \in V \setminus (\{v_1,\ldots,v_k\} \cup L)$ such that
$z \antes v_l$ (if $l \in \mathcal{J}(S)$) and $v_l \antes z$ (if $l \in \mathcal{I}(S)$).
The latter conditions prevent from generating solutions that certainly do not improve the objective function.
Then, consider the solution $S'$ equal to $S$ except for swapping vertices $v_l$ and $z$, i.e.~for $r = 1,\ldots,k$,
\[ v'_r = \begin{cases}
             z, & \textrm{if}~r = l, \\
             v_r, & \textrm{otherwise}.
\end{cases} \]
\item $N'_2$: \emph{Swap two vertices from the sequence}.
Let $l_2 \in \mathcal{I}(S) \cup \mathcal{J}(S)$ and $l_1 < l_2$ such that $v_{l_2} \antes v_{l_1}$ and $v_{l_1} \notin L$.
Then, consider the solution $S'$ equal to $S$ except for swapping $v_{l_1}$ and $v_{l_2}$, i.e.~for $r = 1,\ldots,k$,
\[ v'_r = \begin{cases}
             v_{l_2}, & \textrm{if}~r = l_1, \\
             v_{l_1}, & \textrm{if}~r = l_2, \\
             v_r, & \textrm{otherwise}.
\end{cases} \]
\end{itemize}
\item \emph{Selection of the next solution}. For every $S' \in N'_1(S) \cup N'_2(S)$, $f(S')$ is computed and the solution with the lowest
value is chosen for the next iteration. In case of a tie, the solution is randomly chosen from those ones with the same objective value.
\item \emph{Update of the tabu list and tabu tenure}.
If the new solution comes from the first movement by swapping some $v_l$ with $z$, then $v_l$ is added to $L$. If the new solution
was constructed by swapping some $v_{l_1}$ and $v_{l_2}$, then $v_{l_2}$ is added to $L$.
In both cases, the tenure assigned to the incoming vertex in $L$ is a random integer from the interval $[5,20]$ with uniform distribution.
\end{itemize}
Let $S, S', S''$ be solutions at three consecutive iterations, i.e.~$S' \in N'(S)$ and $S'' \in N'(S')$. It is easy to see that,
regardless of the movement performed, it is not possible to obtain $S'' = S$ due to the tabu mechanism.
The algorithm and, in particular, the fact that $S'' \neq S$ can be illustrated with an example.
Consider the graph of Figure \ref{fig:vi} with $C = V$, and $k = 4$.
Suppose that $L = \emptyset$ and $S_0 = (1,8,2,6)$.
Since 2 is a conflicting vertex, i.e.~$N\la 2\ra \setminus (N\la 1\ra \cup N\la 8\ra) = \emptyset$, this sequence is not legal.
Also, $\mathcal{I}(S_0) = \{3\}$ and $\mathcal{J}(S_0) = \{1,2\}$ (vertices 1 and 8 conflict with 2).
An allowed movement is to swap $v_{l_2} = 2$ with $v_{l_1} = 8$, giving rise to the legal sequence $(1,2,8,6)$, which ends the search.
However, another movement is to swap $v_{l_2} = 2$ with $v_{l_1} = 1$.
Suppose that we perform the latter movement so as to obtain $S_1 = (2,8,1,6)$.
Then, vertex $2$ is marked as tabu.
In the next iteration, $\mathcal{I}(S_1) = \{3\}$ (1 is a conflicting vertex) and $\mathcal{J}(S_1) = \{1,2\}$
(2 and 8 conflict with 1).
At this point, it should be noted that the next solution, $S_2$, can not be equal to $S_0$ because of $2 \in L$ (implying that
the swap between $v_{l_2} = 1$ and $v_{l_1} = 2$ is forbidden).
An allowed movement is to swap $v_l = 8$ with $z = 4$, giving rise to the legal sequence $(2,4,1,6)$, which ends the search.

Below, a pseudocode of the whole algorithm is displayed. There, $\beta$ is a function returning a random real number from $[0,1)$
(the same as in \textsc{Maximalize}).

\medskip

\begin{algorithm}[H] \small
  \SetKwInOut{Input}{Input}
  \SetKwInOut{Output}{Output}
  \Input{An instance $G;C$ and a positive integer $k$.}
  \Output{A maximal legal sequence of size at least $k$ or fail.}
  Initialize $S$ with a random sequence of $k$ different vertices\;
  $L \leftarrow \emptyset$\;
  \While{iteration/time limit is not reached}{
    Compute $\mathcal{I}(S)$\;
    \If{$|\mathcal{I}(S)| = 0$}{Return $\textsc{Maximalize}(S)$ and exit\;}
    Compute $\mathcal{J}(S)$\;
    \For{every $l \in \mathcal{I}(S) \cup \mathcal{J}(S)$}{
      \For{every $z \in V$ not in $S$ or $L$ such that $l \notin \mathcal{I}(S) \lor v_l \antes z$
			                                              and $l \notin \mathcal{J}(S) \lor z \antes v_l$}{
        Compute $|\mathcal{I}(S')| + \beta$ and keep in $S^*$ the solution with the least value\;
			}
		}
    \For{every $l_2 \in \mathcal{I}(S) \cup \mathcal{J}(S)$ such that $l_2 \geq 2$}{
      \For{every $l_1 = 1,\ldots,l_2-1$ such that $v_{l_2} \antes v_{l_1}$ and $v_{l_1} \notin L$}{
        Compute $|\mathcal{I}(S')| + \beta$ and keep in $S^*$ the solution with the least value\;
			}
		}
    \If{no $S^*$ was found (every movement is forbidden)}{Fail\;}
    \For{every $v \in L$}{
      $live(v) \leftarrow live(v) - 1$\;
      \If{$live(v) = 0$}{$L \leftarrow L \setminus \{v\}$\;}
		}
    Let $v_{tabu}$ be $v_l$ if $S^*$ was taken from $N'_1(S)$, and $v_{l_2}$ if $S^*$ was taken from $N'_2(S)$\;
    $L \leftarrow L \cup \{v_{tabu}\}$\;
    Assign a random number from $[5,20]$ to $live(v_{tabu})$\;
    $S \leftarrow S^*$\;
	}
  Fail\;
  \caption{\textsc{TabuSearch}}
\end{algorithm}

\subsection{Improving search with generation of unrelated solutions} \label{TABUALTERNATIVE}

We devised a procedure that provides alternative solutions during the search.
It consists in picking at random some non-conflicting vertices from the current solution $S$ 
to make a subsequence with them (maintaining the order in which they appear in $S$).
Clearly, this subsequence of $S$ is legal. Then, it can be maximalized and, if the resulting sequence
has at least $k$ elements, it is returned as the solution of the search.

The following pseudocode brings details of this procedure, named \textsc{GenAlternativeSol}, and should be added to
\textsc{TabuSearch} before the line that computes $\mathcal{J}(S)$:

\medskip

\begin{algorithm}[H] \small
  \If{$k-|\mathcal{I}(S)| \geq |\mathcal{R}|+1$}{
    \For{$r \in \mathcal{R}$}{
      Let $v_{l_1}, \ldots, v_{l_r}$ be $r$ random non-conflicting vertices from $S$ (i.e.~$l_i\notin\mathcal{I}(S)$).\;
      $S' \leftarrow \textsc{Maximalize}(v_{l_1},\ldots,v_{l_r})$\;
      \If{length of $S'$ is $k$ or greater}{Return $S'$ and exit\;}
		}
	}
  \caption{\textsc{GenAlternativeSol}}
\end{algorithm}

\medskip

\noindent where $\mathcal{R}$ is a set of natural numbers that will be chosen experimentally in the next section.

The CPU time spent by \textsc{GenAlternativeSol} per iteration is negligible and, in most cases, it is able to find legal sequences
of size $k$ unexpectedly, without the need to wait for $f(S)$ to converge to zero.

\section{Computational experiments} \label{SSCOMPU}

This section is devoted to present computational experiments in order to answer several questions:
which formulation performs better?, does the addition of inequalities \eqref{SUPERNOVA} as cuts improve the performance?,
how effective is the tabu search to find good solutions?, how large are the instances that our approach can tackle?

The experiments have been carried out over several random and benchmark instances.
A computer equipped with an Intel i7-7700 3.6GHz CPU, 8Gb of RAM, and IBM ILOG CPLEX 12.7
has been used. Each run has been performed on one thread of the CPU.

Random instances are generated as follows. For given numbers $n \in \mathbb{Z}_+$ and $p \in [0, 1]$, a graph is generated by
starting from the empty graph of $n$ vertices and adding edges with probability $p$. For example, if $p = 1$, then a complete graph is obtained.
It is expected that the resulting graph has an edge density similar to $p$, so graphs with $p = 0.1$ and $0.9$ are referred to as \emph{low}
and \emph{high} density respectively, while those ones with $p \in \{0.3, 0.5, 0.7\}$ are referred to as \emph{medium} density.

Our implementation as well as all the instances can be downloaded from:
\begin{center}
\texttt{https://www.fceia.unr.edu.ar/$\sim$daniel/stuff/grundy.zip}
\end{center}

\subsection{Comparing formulations}

In this experiment, we evaluate the 8 formulations presented in Subsection~\ref{FORMLIST} over random instances obtained as follows:
for each $n \in \{10, 20, 30\}$ and $p \in \{0.1, 0.3, 0.5, 0.7, 0.9\}$, 3 graphs ($G_1, G_2, G_3$) with $n$ vertices and edge probability $p$
are generated, and for each $i \in \{1,2,3\}$, two instances are considered: $G_i;V$ (Grundy domination number of $G_i$) and $G_i;\emptyset$
(Grundy total domination number of $G_i$), giving rise to 90 random instances.

Each run consists of invoking \textsc{GetInitialBounds} and solving one of the 8 formulations
via a pure branch-and-bound (more precisely, the MIP optimizer of CPLEX is used, with heuristics, cuts and presolve turned off).
At the beginning of the optimization, the initial solution is injected as the first incumbent. 
Neither inequalities \eqref{NONOPTIMALREMOVE} nor \eqref{SUPERNOVA} are considered in this experiment.
A time limit of two hours is imposed for each run.

Table \ref{tab:1} reports the results. Each row presents information about the 6 instances with the same number of vertices $n$ and
probability $p$, given in the first and second columns. The third column (RG$_{initial}$) shows the average of initial percentage relative gaps 
(i.e.~the value $100(UB-LB)/LB$ where $LB$ and $UB$ are the bounds provided by \textsc{GetInitialBounds}) over the 6 instances. The next columns report the following values:
``RelGap'' refers to the average over the 6 instances
of percentage relative gaps at the end of the optimization, ``Solved'' is the number of instances solved within the established time limit,
and ``Time'' is the average over the solved instances of time elapsed during optimization (in seconds). In case none of the 6 instances are
solved, a mark ``$-$'' is displayed. Some values mentioned in the analysis are highlighted in boldface.

None of the high density instances are reported since they are satisfactorily solved by \textsc{GetInitialBounds}.\\

\noindent \emph{Analysis}. At first glance, we can observe that the harder instances are those with 30\% of density.
In fact, the lower the density of the graph is the longer the sizes of maximal legal sequences are, and thus, more variables the models have.
Regarding the order of graphs, for 10 vertices the resolution is straightforward, but for 20 and 30 vertices, there are instances that cannot
be solved in two hours of CPU time. In general, the formulations could close or significantly reduce the initial gaps of instances up to 20
vertices or more than 50\% of edge density.
On the other hand, we did not perceive any tendency between instances with $V=C$ and $V=\emptyset$.

For $n = 20$ and $p = 0.3$, $F_3$ and $F_4$ solve 2 of 6 instances and both formulations report the smallest gap (and $F_3$ uses
7,5\% less time than $F_4$, a small difference).
For $n = 20$ and $p = 0.5$, $F_4$ and $F_6$ solve all the instances between (roughly) 2 and 3 times faster than the others.
For $n = 30$ and $p = 0.5$, $F_3$, $F_4$ and $F_8$ solve almost all instances (compared to the other formulations) and, in particular, $F_4$
presents the smallest gap and elapsed time.
For $n = 30$ and $p = 0.7$, $F_4$ and $F_8$ solve all the instances between (roughly) 2 and 5 times faster than the others.
By taking these facts into account, we conclude that $F_4$ performs better than the others on average.

\subsection{Reinforcing the relaxations} \label{REINFRELAX}

As we have pointed out in Subsection \ref{VALIDINEQSECTION}, the addition of violated valid inequalities to the relaxations can improve
the performance of the solver.
Our cutting-plane algorithm consists of the separation of inequalities \eqref{SUPERNOVA}, specifically Type I and II
(see Subsection \ref{VALIDINEQSECTION}).
Both have its own routine which is invoked after a linear relaxation is solved. If at least one cut is
generated, it is added to the relaxation and the latter is reoptimized.
In particular, the separation of Type I inequalities is performed 10 times in the root node (that means at most 10 reoptimizations), twice in
nodes with depths 1 and 2, and once in nodes with depths 3 to 10.
The routine that separates Type II inequalities is executed after the one for Type I in nodes with depth at most 5.
Below, we describe the implementation of both routines. The current fractional solution is denoted by $(x^*, y^*)$.

\begin{itemize}
\item \emph{Separation of Type I inequalities}.
Before starting the optimization, create sets
$$\mathcal{W}_u = \{ w \in N\la u\ra : |N\la w\ra| \geq 2 ~\textrm{and, for all}~ v \in N\la u\ra \setminus \{w\}, w \antes v ~\textrm{and}~ v \antes w\}$$
for each $u \in V$.
Each time the separation routine is invoked, assign $\mathcal{A} \leftarrow V$.
Then, for every $u \in V$ and $w \in \mathcal{W}_u \cap \mathcal{A}$ do the following. Set $sum \leftarrow y^*_{w1}$.
For all $i = 2,\ldots,m$, do $sum \leftarrow sum + y_{wi}$ and check whether $x^*_{ui} + sum > 1.1$. In that case, add
$x_{ui} + \sum_{j=1}^i y_{w j} \leq 1$ as a cut and remove $w$ from $\mathcal{A}$.

Set $\mathcal{A}$ stores those vertices  ``$w$'' not used by cuts from previous iterations, thus preventing the generation of
cuts with similar support.
\item \emph{Separation of Type II inequalities}.
Before starting the optimization, create sets
\begin{multline*}
  \mathcal{W}_{u_1 u_2} = \{ w \in \mathcal{W}_{u_1} \cap \mathcal{W}_{u_2} :
     ~\textrm{there are}~ z_1 \in N\la u_1\ra \setminus N\la u_2\ra, z_2 \in N\la u_2\ra \setminus N\la u_1\ra\\
		\textrm{such that}~ N\la w\ra \setminus (\{u_2\} \cup N\la z_1\ra) \neq \emptyset,
		N\la w\ra \setminus (\{u_1\} \cup N\la z_2\ra) \neq \emptyset\}
\end{multline*}
for each pair $\{u_1, u_2\} \subseteq V$.
The separation routine is executed immediately after the separation of Type I inequalities and makes use of the vertices that remain in
$\mathcal{A}$. 
For every $\{u_1, u_2\} \subseteq V$ and $w \in \mathcal{W}_{u_1 u_2} \cap \mathcal{A}$ do the following. Set $sum \leftarrow y^*_{w1}$.
For all $i = 2,\ldots,m$, do $sum \leftarrow sum + y_{wi}$ and check whether $x^*_{u_1 i} \notin \mathbb{Z}$ and
$x^*_{u_2 i} \notin \mathbb{Z}$. In that case, for all $k = 1,\ldots,i$, if $y^*_{wk} \notin \mathbb{Z}$, then check whether
$x^*_{u_1 i} + x^*_{u_2 i} + sum + \sum_{v \in N^{\cup}} y_{vk} > 2.2$ and, in that case, add
$x_{u_1 i} + x_{u_2 i} + \sum_{j=1}^i y_{wj} + \sum_{v \in N^{\cup}} y_{vk} \leq 2$ as a cut and remove $w$ from $\mathcal{A}$.
\end{itemize}
These routines have been designed in a previous work (an extended abstract) where some polyhedral aspects, such as the dimension of the face defining the inequality, have been taken into account \cite{LAGOS2017}.\\

In Subsection \ref{NONOPTIMALSECTION}, we presented a set of equalities that can be added to the formulation. We propose two treatments of these
equalities:
\begin{itemize}
\item \emph{Addition of equalities at the beginning}. Before starting the optimization (i.e.~when the model is populated in the
memory of CPLEX), add \eqref{NONOPTIMALREMOVE} for $i=1,\ldots,LB$, where $LB$ is the lower bound generated by \textsc{GetInitialBounds}. Also, do not add neither \eqref{RESTR1} for $i=1,\ldots,LB$ nor \eqref{RESTR8NEW} for $i=2,\ldots,LB$.
\item \emph{Treatment of equalities as cuts}. The following routine is executed after a linear relaxation is solved.
Let  $(x^*, y^*)$ be the current fractional solution and $z \in \mathbb{Z}$ be the objective function value of the best integer solution found
so far. In other words, $z$ is the best available lower bound.
For all $i = z, \ldots, 1$, check whether $\sum_{v \in V} y^*_{vi} < 0.9$ and, in that case, add $\sum_{v \in V} y_{vi} = 1$ as a cut. Otherwise, exit the loop.
\end{itemize}
Preliminary experiments show that the first approach is better, possibly because the lower bound is infrequently improved.
Actually, the initial lower bound is usually already close to the optimum.
From now on, when we refer to equalities \eqref{NONOPTIMALREMOVE}, we consider they are added at the beginning of the optimization.

In the next experiment, we evaluate the presence of inequalities \eqref{SUPERNOVA} (types I and II) and equalities \eqref{NONOPTIMALREMOVE}
during the optimization, over the same set of instances of the previous experiment.
We consider the following 6 variants: Base (i.e.~no cuts), $+T_1$ (i.e.~with cuts of type I), $+T_1+T_2$ (i.e.~with both type of cuts),
$+\eqref{NONOPTIMALREMOVE}$, $+\eqref{NONOPTIMALREMOVE}+T_1$, and $+\eqref{NONOPTIMALREMOVE}+T_1+T_2$.
Formulation $F_4$ is used in all cases.
Each run consists in invoking \textsc{GetInitialBounds} and solving one of these 6 variants.
A time limit of two hours is imposed. Again, CPLEX cuts, heuristics and presolve are turned off.

Table \ref{tab:2} reports the results in the same format as Table \ref{tab:1}. Instances of 10 vertices have been omitted as they are
too easy for all the variants.\\

\noindent \emph{Analysis}. For $n = 20$ and $p = 0.1$, the addition of equalities reduces the average of CPU time to the half and presents
a little improvement in the relative gap, and for $p = 0.3$, it is also able to solve 2 more instances. In particular, the best variant
is $+\eqref{NONOPTIMALREMOVE}+T_1+T_2$, and in second place, $+\eqref{NONOPTIMALREMOVE}+T_1$.
Cuts lose effect in the densest instances, although these ones are also easier to solve.
For $n = 30$ and $p = 0.5$, $+\eqref{NONOPTIMALREMOVE}+T_1$ performs better as it solves all the instances. Again, cuts lose effect
for instances of high density ($n = 30$ and $p = 0.9$).

Not always adding inequalities of Type II leads to a general improvement in dual bound. For example, there is an instance
with $n = 30$, $p = 0.5$ (precisely $G_3$ with $C = \emptyset$) that $\eqref{NONOPTIMALREMOVE}+T_1$ solves but $\eqref{NONOPTIMALREMOVE}+T_1+T_2$ does not.
In the first case, 1631 cuts of Type I are produced along the optimization and, in the second one, 570 cuts of
Type I and 32 cuts of Type II are generated. It seems that the introduction of the latter diminished the generation of Type I cuts here.

From these computational experiments, we propose to add \eqref{NONOPTIMALREMOVE} as well as to enable cuts of type I when the density of the graph is less than 60\%, and
enable cuts of type II when the density is less than 40\%.

\subsection{Determination of set $\mathcal{R}$}

In Subsection \ref{TABUALTERNATIVE}, we present a procedure that eventually provides legal sequences during the tabu search.
These sequences are generated by maximalizing a subsequence of $r$ non-conflicting vertices taken from the current solution.
Here, we carry out an experiment in order to determine which values of $r$ yield the longest sequences.

Tabu search is executed over random instances of 100 and 200 vertices (and $C=V$) for one hour, and with a given initial $k$.
\textsc{GenAlternativeSol} is implemented with $\mathcal{R} = \{2,3,\ldots,9\}$ but, each time it finds a sequence
of size $k$ from another of size $r \in \mathcal{R}$, a counter associated to $r$ is incremented by one unit and that solution is discarded.
The values of $k$ were chosen so that there exists a sequence of size $k$ but the tabu search does not find any within one hour of time.
For this reason, only graphs with density up to 50\% were considered, since it seems that the tabu search quickly converges to the optimal
solution for higher densities.

Figure \ref{figtab:1} shows 6 histogram-like charts, each one corresponding to a graph with $n \in \{100,200\}$ and $p \in \{0.1, 0.3, 0.5\}$.
The value of $k$ is also reported.
For each $r \in \mathcal{R}$, a bar is drawn along with the counter for $r$, i.e.~the number of times \textsc{GenAlternativeSol}
reaches a sequence of size $k$ from another of size $r$.\\

\noindent \emph{Analysis}. The values of $r$ that yield the longest sequences are 2, 3 and 4. Although $r = 2$ is the best in most cases,
$r = 3$ is better for the hardest case (sequences of size 78, the longest one). Note also that $r=4$ yields sequences of size $k$ in all
tested instances, but this fact does not happen for $r \geq 5$. We conclude that $\mathcal{R} = \{2,3,4\}$ is a reasonable setting for
\textsc{GenAlternativeSol}.\\

In order to find out how much this procedure improves the search of sequences, we run two versions of the tabu search: one with
\textsc{GenAlternativeSol} enabled (with $\mathcal{R} = \{2,3,4\}$), and the other, disabled.
Same instances as before are used, plus others with $n \in \{100,200\}$ and $p \in \{0.7, 0.9\}$.
Due to the non-deterministic nature of the algorithm, three runs per instance are performed.
Each run starts with $k = 3$ and each time a sequence of size $k$ is found, the tabu search is restarted with $k+1$.
Table \ref{tab:3} shows the best $k$ achieved by each run and the time the algorithm took to reach it in brackets.
Best run is highlighted in boldface.\\

\noindent \emph{Analysis}. \textsc{GenAlternativeSol} dramatically improves the search of long legal sequences.
In a matter of seconds, it finds solutions that are not possible to obtain otherwise within one hour of execution. From now on, it is
enabled by default.

\subsection{Limits of our exact algorithm}

The goal of this experiment is to estimate the largest size of an instance that can be solved in a fixed amount of time (four hours).
In each run, \textsc{GetInitialBounds} and the tabu search are invoked with a total time limit of 30 seconds, and a maximum of 50000
iterations for the tabu search. Then, formulation $F_4$ is solved.
Equalities \eqref{NONOPTIMALREMOVE} and cuts are added/enabled according to the criterion given in Subection \ref{REINFRELAX}.
In order to differentiate the phases of the algorithm, we call \emph{initial phase} the search for initial bounds
(\textsc{GetInitialBounds} and the tabu search) and \emph{optimization phase} the resolution of the integer formulation.

The following instances are considered:
\begin{itemize}
\item \emph{Graphs from the DIMACS challenge} (\texttt{https://mat.gsia.cmu.edu/COLOR04}).
It is a standard set of benchmark instances which were originally selected for testing graph coloring algorithms, although later
it was used for other optimization problems in graphs, in particular dominating set problems \cite{CHALUPA2018}.
We consider those graphs up to 50 vertices, and their complements (names are suffixed with letter \emph{c} to identify them).
For each graph, we create two instances: one with $C=V$ and the other with $C=\emptyset$, giving a total amount of 32 instances.
\item \emph{Random instances}. We consider graphs of 25 and 50 vertices, with edge probability $p \in \{0.1, 0.3, 0.5, 0.7, 0.9\}$
and $C \in \{V, \emptyset\}$, giving a total amount of 20 instances. Each instance is identified by G$n\_p$.
\item \emph{Real instances} (\texttt{https://www.buenosaires.gob.ar/laciudad/barrios}).
We consider two instances based on the map of neighborhoods of the city of Buenos Aires, see Figure \ref{fig:bsas}.
In the first one (\emph{full version}), each district corresponds to a neighborhood, giving an amount of 48.
In the second one (\emph{small version}), each district is associated to one or more neighborhoods, depending on the total area
(e.g.~Villa Ortúzar, Parque Chas and other small neighborhoods are gathered into one district), giving an amount of 21.
Table \ref{tab:6} gives the districts considered in each instance, and the best solutions found by our algorithm (expressed by the allocation order of the companies). 
In particular, the solution for the second case is optimal.
\end{itemize}

Table \ref{tab:5} reports the results. Each row corresponds to two instances, one with $C=V$ and the other with $C=\emptyset$.
The first, second and third columns report the name of the instance, the number of vertices and the edge density.
The next columns give the best upper/lower bounds obtained at the beginning and at the end of the optimization. A star ``*'' informs that the initial upper bound is provided by the user.
A mark ``$-$'' is displayed if the time limit is reached.
A dagger ``$\dagger$'' indicates that the tabu search improves the solution given by \textsc{GetInitialBounds}.
Values in boldface reveal an improvement during the optimization phase. If the optimality is reached by the initial phase,
final bounds are not reported.\\

\noindent \emph{Analysis}. Despite the short time allocated to the initial phase, it is very effective, mainly on high density instances.
It is able to solve 12 instances (out of 54) in a matter of milliseconds.
Also, for every solved instance, the initial phase actually provides the optimal solution.
In particular, the tabu search performs very well: such instances where \textsc{GetInitialBounds} does not provide the optimal solution,
the latter is delivered by the tabu search; also it improves the solution generated by \textsc{GetInitialBounds} on several hard instances
such as low-density graphs of 50 vertices.

Regarding the optimization phase, our approach is able to exactly solve more than one-third of the instances (16 out of 42)
and to decrease the upper bound (see values in boldface) in half of the cases. 
A limit of our approach seems to be based on the initial upper bound. It is unlikely that the instance could be solved for $UB \geq 20$
in four hours.
On the other hand, instances with $UB \leq 6$ are easily solved regardless of the size of the graph.\\

As the initial solutions are optimal or near the optimal solution, it is natural to propose the following procedure.
Instead of using the upper bound provided by the initial phase, set $UB \leftarrow LB+1$.
Therefore, if the optimization finishes with objective value equal to the initial $LB$, then the initial solution is optimal.
In other words, we are using the solver to \emph{decide} whether $\grd(G;C) \geq LB+1$ or not.
We evaluated with this procedure those instances where there is a gap of at least two units between $UB$ and $LB$, and we
were able to solve one more instance: $\grd(\textrm{queen6\_6c},\emptyset) = 8$.

\subsection{Improving bounds of $\grd(K_{n,r};V)$}

For given $r$, $n$ positive integers such that $n \geq 2r$, the \emph{Kneser graph} $K_{n,r}$ is defined as follows. The vertex set represents all subsets of $\{1,\ldots,n\}$ with $r$ elements, and two vertices are adjacent if and only if the corresponding subsets are disjoint.
Since, for $r = 1$, this graph is isomorphic to a complete one and, for $n = 2r$, it is isomorphic to a disjoint union of ${2r\choose r}/2$ edges,
we assume that $r \geq 2$ and $n \geq 2r+1$.

Due to its structure and properties, there is interest in knowing the value of different graph parameters of the Kneser graph
(see, for example, the famous Lov\'asz's proof of Kneser conjecture \cite{LOVASZKNESER}).
In particular, in a recent work \cite{GRUNDYKNESER}, the authors give the Grundy total domination number of the Kneser graph:
$\grd(K_{n,r};\emptyset) = {2r\choose r}$.
The Grundy domination number is, however, partially characterized.
They prove that if $n$ is large enough, then $\grd(K_{n,r};V)$ coincides with the independence number of $K_{n,r}$:
\begin{thm} \label{KNESERRESULT} \cite{GRUNDYKNESER} 
For any $r \geq 2$, there exists $\bar{n}_r \in \mathbb{Z}_+$ such that $\grd(K_{n,r};V) = {n - 1\choose r - 1}$ for any $n \geq \bar{n}_r$.
In particular, $\bar{n}_2 = 6$.
\end{thm}
They also compute the following cases: $\grd(K_{5,2};V) = 5$ and $\grd(K_{7,3};V) = 20$.
However, for $r \geq 3$, $\bar{n}_r$ remains unknown and the best bounds in the literature are:
\[ {n - 1\choose r - 1} \leq \grd(K_{n,r};V) \leq {n\choose r}-{n - r\choose r}\]
The lower bound comes from the length of a sequence whose set of vertices is a maximum independent set of $K_{n,r}$, while the upper bound
is given by Prop.~2.1 of \cite{BRESAR2014} and is equal to $m_1$ (defined in Subsection \ref{PROPERTIESGGDP}).

This subsection intends to improve these bounds for some cases, specifically Kneser graphs with up to 800 vertices.
Table \ref{tab:4} reports the results obtained by invoking \textsc{GetInitialBounds} and just after the tabu search, for an hour of CPU time.
The first three columns report the parameters $n$ and $r$, and the number of vertices. The next two columns show the value
${n\choose r}-{n - r\choose r}$ and $m_3$ (given by \textsc{GetInitialBounds}). The last columns have ${n - 1\choose r - 1}$
and the length of the best legal sequence found together with the time it took to reach such a sequence in brackets.
A dagger ($\dagger$) indicates that the tabu search improves the solution given by \textsc{GetInitialBounds}.
Best values are displayed in boldface.\\

\noindent \emph{Analysis}. Although it is expected that $m_3 \leq m_1$ (see Subsection \ref{PROPERTIESGGDP}), in these instances $m_3$ is
\emph{strictly} less. Thus, our procedure provides better upper bounds. Besides, our approach also computes better lower bounds for 
$K_{8,3}$, $K_{9,4}$, $K_{10,4}$, $K_{11,4}$, $K_{11,5}$ and $K_{12,5}$, despite the fact that it does not exploit any
particular characteristic of the structure of these graphs.
In particular, the best sequence is achieved by the tabu search on three of the hardest instances ($K_{10,4}$, $K_{11,4}$ and $K_{12,5}$).

Note that our approach is not able to provide a legal sequence larger than ${n - 1\choose r - 1}$ for $r = 3$ and $n \geq 9$.
The same happens for $r = 4$ and $n \geq 12$. We believe that ${n - 1\choose r - 1}$ is the optimal solution for these cases and
we conjecture that $\bar{n}_r = 3r$ for any $r$ in Theorem \ref{KNESERRESULT}.

\subsection{Tabu search on large instances}

In the experiment performed in the previous subsection, we observed that \textsc{GenInitialBounds} spends a considerable amount of
time, e.g.~for $K_{12,4}$ (roughly 500 vertices) it takes 170 seconds and for $K_{12,5}$ (roughly 800 vertices) it takes 613 seconds.
For graphs with more than 1000 vertices, this heuristic becomes impractical, e.g.~for $K_{20,3}$ (1140 vertices) it takes 11500 seconds.
This behavior is expected since its time complexity is cubic on the number of vertices of the graph.

However, the tabu search (without the initial solution given by \textsc{GenInitialBounds}) can still generate good solutions for
large instances.
In this last experiment, we consider some instances up to 10000 vertices where the optimal parameter is already known:
\begin{itemize}
\item Grundy total domination on Kneser graphs, $\grd(K_{n,r};\emptyset) = {2r\choose r}$ \cite{GRUNDYKNESER},
\item Grundy domination on Web graphs, $\grd(W_n^r;V) = n - 2r$ (Prop.~\ref{WebProof}).
\end{itemize}
We set $UB$ with the optimum value and execute the tabu search starting from $k = 3$ and do not stop until a legal sequence of size $UB$
is reached.
Time in seconds for each instance is reported in Table \ref{tab:7}. Observe that, in all cases, the tabu search is able to find the optimal
solution within one hour of CPU time.

\section{Conclusions} \label{SSCONCLU}

In this work, an optimization problem that generalizes the Grundy domination and Grundy total domination problems is introduced.
Some properties of this problem and the exact value of the parameter $\grd(G;C)$ for two families of graphs are given.
This problem is modeled as an integer linear program. Some additional families of constraints are considered in order to
provide different formulations of the same model. The validity of another family of inequalities, \eqref{SUPERNOVA}, is proved.
Since these inequalities are very generic and its number is exponential, two subfamilies (named Type I and II) are considered and
polynomial-time routines for separating them are detailed.

On the other hand, a greedy heuristic, \textsc{GenInitialBounds}, is proposed. It provides an initial legal sequence and an
upper bound of $\grd(G;C)$ which, for some high density graphs $G$, is able to certify the optimality of the obtained sequence.
The size of the initial legal sequence can be further improved by a tabu search. This algorithm dramatically improves its performance
when an additional mechanism, \textsc{GenAlternativeSol}, is added. 
For all instances where optimality could be proved, the tabu search with \textsc{GenAlternativeSol} was able to find the optimal solution within one hour.
Experiments give evidence that this approach yields optimal or near-optimal solutions for instances up to 10000 vertices.

Our exact approach, i.e.~\textsc{GenInitialBounds} plus the tabu search with \textsc{GenAlternativeSol} plus the optimization of one of 
the formulations with the aid of a cutting-plane algorithm that separates Type I and II inequalities, can exactly solve instances ranging from
20 to 50 vertices depending on the edge density of the graph (this includes one of two real-life instances).

Besides the computational results, the resolution of the Grundy domination problem on Kneser graphs of several sizes allows us to
state a theoretical conjecture: that for any $r \geq 2$ and $n \geq 3r$, $\grd(K_{n,r};V) = {n - 1\choose r - 1}$.


\newpage

\begin{table}[t] \centering \scriptsize
\begin{tabular}{@{\hspace{3pt}}c@{\hspace{3pt}}@{\hspace{3pt}}c@{\hspace{3pt}}@{\hspace{3pt}}c@{\hspace{3pt}}@{\hspace{3pt}}c@{\hspace{3pt}}@{\hspace{3pt}}c@{\hspace{3pt}}@{\hspace{3pt}}c@{\hspace{3pt}}@{\hspace{3pt}}c@{\hspace{3pt}}@{\hspace{3pt}}c@{\hspace{3pt}}@{\hspace{3pt}}c@{\hspace{3pt}}@{\hspace{3pt}}c@{\hspace{3pt}}@{\hspace{3pt}}c@{\hspace{3pt}}@{\hspace{3pt}}c@{\hspace{3pt}}}
\hline
 $n$ & $p$ & RG$_{initial}$ & Param. & $F_1$ & $F_2$ & $F_3$ & $F_4$ & $F_5$ & $F_6$ & $F_7$ & $F_8$ \\
\hline
 &  &  & RelGap & 0.00 & 0.00 & 0.00 & 0.00 & 0.00 & 0.00 & 0.00 & 0.00 \\ 
 10 & 0.1 & 20.73 & Solved & 6 & 6 & 6 & 6 & 6 & 6 & 6 & 6 \\ 
 &  &  & Time & 0.9 & 0.8 & 0.5 & 0.4 & 0.8 & 0.8 & 0.4 & 0.4 \\      \cmidrule{2-12}
 &  &  & RelGap & 0.00 & 0.00 & 0.00 & 0.00 & 0.00 & 0.00 & 0.00 & 0.00 \\ 
 10 & 0.3 & 14.88 & Solved & 6 & 6 & 6 & 6 & 6 & 6 & 6 & 6 \\ 
 &  &  & Time & 1.2 & 0.9 & 1.0 & 1.1 & 1.2 & 1.1 & 1.4 & 1.4 \\      \cmidrule{2-12}
 &  &  & RelGap & 0.00 & 0.00 & 0.00 & 0.00 & 0.00 & 0.00 & 0.00 & 0.00 \\ 
 10 & 0.5 & 13.75 & Solved & 6 & 6 & 6 & 6 & 6 & 6 & 6 & 6 \\ 
 &  &  & Time & 0.7 & 0.6 & 0.6 & 0.6 & 0.7 & 0.6 & 0.6 & 0.5 \\      \cmidrule{2-12}
 &  &  & RelGap & 0.00 & 0.00 & 0.00 & 0.00 & 0.00 & 0.00 & 0.00 & 0.00 \\ 
 10 & 0.7 & 2.78 & Solved & 6 & 6 & 6 & 6 & 6 & 6 & 6 & 6 \\ 
 &  &  & Time & 0.1 & 0.1 & 0.1 & 0.1 & 0.1 & 0.1 & 0.1 & 0.1 \\     
\hline
 &  &  & RelGap & 4.98 & 4.98 & 6.01 & 6.01 & 4.98 & 4.98 & 7.05 & 6.01 \\ 
 20 & 0.1 & 20.21 & Solved & 3 & 3 & 3 & 3 & 3 & 3 & 3 & 3 \\ 
 &  &  & Time & 1113.8 & 1444.5 & 1483.5 & 2049.8 & 1286.2 & 1211.5 & 1283.7 & 1696.7 \\      \cmidrule{2-12}
 &  &  & RelGap & 6.93 & 12.06 & {\bf 6.84} & {\bf 6.84} & 11.74 & 9.47 & 10.81 & 12.20 \\ 
 20 & 0.3 & 18.15 & Solved & 1 & 0 & {\bf 2} & {\bf 2} & 0 & 1 & 0 & 0 \\ 
 &  &  & Time & 5231.9 & $-$ & 4650.8 & 5026.7 & $-$ & 4224.4 & $-$ & $-$ \\      \cmidrule{2-12}
 &  &  & RelGap & 0.00 & 0.00 & 0.00 & 0.00 & 0.00 & 0.00 & 0.00 & 0.00 \\ 
 20 & 0.5 & 19.41 & Solved & 6 & 6 & 6 & 6 & 6 & 6 & 6 & 6 \\ 
 &  &  & Time & 289.4 & 356.6 & 272.2 & {\bf 98.4} & 280.1 & {\bf 97.4} & 312.5 & 183.5 \\      \cmidrule{2-12}
 &  &  & RelGap & 0.00 & 0.00 & 0.00 & 0.00 & 0.00 & 0.00 & 0.00 & 0.00 \\ 
 20 & 0.7 & 16.13 & Solved & 6 & 6 & 6 & 6 & 6 & 6 & 6 & 6 \\ 
 &  &  & Time & 18.8 & 10.0 & 11.9 & 5.5 & 11.7 & 10.4 & 15.3 & 11.4 \\     
\hline
 &  &  & RelGap & 18.02 & 17.99 & 18.00 & 18.04 & 18.01 & 17.71 & 18.04 & 18.04 \\ 
 30 & 0.1 & 18.04 & Solved & 0 & 0 & 0 & 0 & 0 & 0 & 0 & 0 \\ 
 &  &  & Time & $-$ & $-$ & $-$ & $-$ & $-$ & $-$ & $-$ & $-$ \\      \cmidrule{2-12}
 &  &  & RelGap & 24.42 & 24.42 & 24.42 & 24.42 & 24.42 & 24.42 & 24.42 & 24.42 \\ 
 30 & 0.3 & 24.42 & Solved & 0 & 0 & 0 & 0 & 0 & 0 & 0 & 0 \\ 
 &  &  & Time & $-$ & $-$ & $-$ & $-$ & $-$ & $-$ & $-$ & $-$ \\      \cmidrule{2-12}
 &  &  & RelGap & 9.60 & 7.97 & 3.03 & {\bf 1.52} & 9.18 & 5.98 & 8.06 & 3.03 \\ 
 30 & 0.5 & 16.77 & Solved & 2 & 2 & {\bf 5} & {\bf 5} & 2 & 3 & 2 & {\bf 5} \\ 
 &  &  & Time & 4241.4 & 3871.1 & 2964.8 & {\bf 2480.2} & 5305.1 & 3541.5 & 4305.5 & 2722.3 \\      \cmidrule{2-12}
 &  &  & RelGap & 0.00 & 0.00 & 0.00 & 0.00 & 0.00 & 0.00 & 0.00 & 0.00 \\ 
 30 & 0.7 & 18.65 & Solved & 6 & 6 & 6 & 6 & 6 & 6 & 6 & 6 \\ 
 &  &  & Time & 254.9 & 218.2 & 197.2 & {\bf 110.5} & 384.9 & 628.6 & 205.5 & {\bf 116.9} \\ 
\end{tabular}
\caption{Comparison of formulations}  \label{tab:1}
\end{table}

\begin{table}[t] \centering \scriptsize
\begin{tabular}{@{\hspace{3pt}}c@{\hspace{3pt}}@{\hspace{3pt}}c@{\hspace{3pt}}@{\hspace{3pt}}c@{\hspace{3pt}}@{\hspace{3pt}}c@{\hspace{3pt}}@{\hspace{3pt}}c@{\hspace{3pt}}@{\hspace{3pt}}c@{\hspace{3pt}}@{\hspace{3pt}}c@{\hspace{3pt}}@{\hspace{3pt}}c@{\hspace{3pt}}@{\hspace{3pt}}c@{\hspace{3pt}}@{\hspace{3pt}}c@{\hspace{3pt}}}
\hline
     &     &                &        & \multicolumn{3}{c}{$F_4$} &  \multicolumn{3}{c}{$F_4$+\eqref{NONOPTIMALREMOVE}} \\
 $n$ & $p$ & RG$_{initial}$ & Param. & Base & +$T_1$ & +$T_1$+$T_2$ & Base & +$T_1$ & +$T_1$+$T_2$ \\
\hline
 &  &  & RelGap & 6.01 & 4.98 & 4.98 & 4.98 & 4.98 & 4.98 \\ 
 20 & 0.1 & 20.21 & Solved & 3 & 3 & 3 & 3 & 3 & 3 \\ 
 &  &  & Time & 2049.8 & 1961.0 & 1976.8 & {\bf 993.6} & {\bf 981.7} & {\bf 969.2} \\      \cmidrule{2-10}
 &  &  & RelGap & 6.84 & 7.04 & 5.80 & 8.23 & 4.96 & {\bf 2.58} \\ 
 20 & 0.3 & 18.15 & Solved & 2 & 2 & 2 & 1 & 3 & {\bf 4} \\ 
 &  &  & Time & 5026.7 & 2495.1 & 1442.1 & 3025.4 & 4385.8 & 2572.6 \\      \cmidrule{2-10}
 &  &  & RelGap & 0.00 & 0.00 & 0.00 & 0.00 & 0.00 & 0.00 \\ 
 20 & 0.5 & 19.41 & Solved & 6 & 6 & 6 & 6 & 6 & 6 \\ 
 &  &  & Time & 98.4 & 80.7 & 137.9 & 206.4 & 101.5 & 99.5 \\      \cmidrule{2-10}
 &  &  & RelGap & 0.00 & 0.00 & 0.00 & 0.00 & 0.00 & 0.00 \\ 
 20 & 0.7 & 16.13 & Solved & 6 & 6 & 6 & 6 & 6 & 6 \\ 
 &  &  & Time & 5.5 & 8.5 & 7.2 & 16.4 & 12.8 & 15.4 \\     
\hline
 &  &  & RelGap & 18.04 & 16.53 & 17.98 & 18.04 & 18.04 & 18.04 \\ 
 30 & 0.1 & 18.04 & Solved & 0 & 0 & 0 & 0 & 0 & 0 \\ 
 &  &  & Time & $-$ & $-$ & $-$ & $-$ & $-$ & $-$ \\      \cmidrule{2-10}
 &  &  & RelGap & 24.42 & 24.42 & 24.42 & 24.42 & 24.42 & 24.42 \\ 
 30 & 0.3 & 24.42 & Solved & 0 & 0 & 0 & 0 & 0 & 0 \\ 
 &  &  & Time & $-$ & $-$ & $-$ & $-$ & $-$ & $-$ \\      \cmidrule{2-10}
 &  &  & RelGap & 1.52 & 6.36 & 6.36 & 8.03 & {\bf 0.00} & 3.33 \\ 
 30 & 0.5 & 16.77 & Solved & 5 & 4 & 4 & 3 & {\bf 6} & 5 \\ 
 &  &  & Time & 2480.2 & 2230.4 & 1972.9 & 2304.1 & 3116.4 & 3702.7 \\      \cmidrule{2-10}
 &  &  & RelGap & 0.00 & 0.00 & 0.00 & 0.00 & 0.00 & 0.00 \\ 
 30 & 0.7 & 18.65 & Solved & 6 & 6 & 6 & 6 & 6 & 6 \\ 
 &  &  & Time & {\bf 110.5} & 650.3 & 178.7 & 460.9 & 218.3 & 172.1 \\      \cmidrule{2-10}
\end{tabular}
\caption{Strengthening the relaxation}  \label{tab:2}
\end{table}

\begin{figure}[t] \centering \scriptsize
\includegraphics[width=1\textwidth]{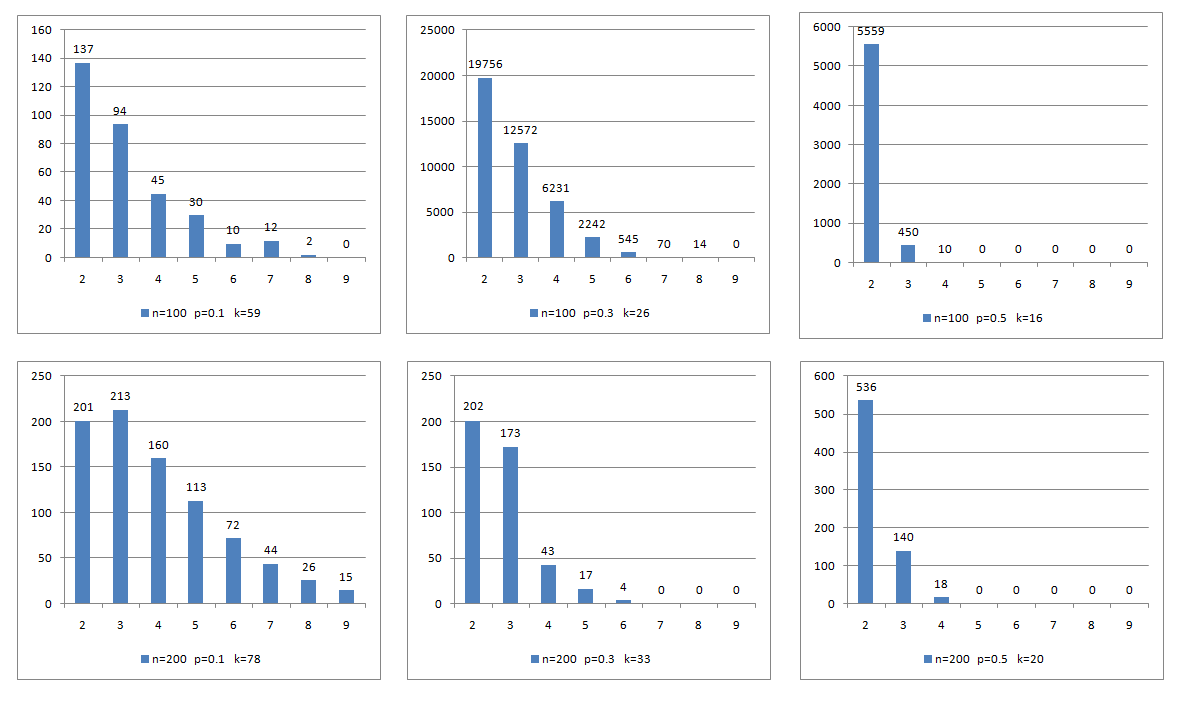}
\caption{Number of times a sequence of size $k$ is reached from another of size $r$}  \label{figtab:1}
\end{figure}

\begin{table}[t] \centering \small
\begin{tabular}{@{\hspace{3pt}}c@{\hspace{3pt}}@{\hspace{3pt}}c@{\hspace{3pt}}@{\hspace{3pt}}c@{\hspace{3pt}}@{\hspace{3pt}}c@{\hspace{3pt}}@{\hspace{3pt}}c@{\hspace{3pt}}@{\hspace{3pt}}c@{\hspace{3pt}}@{\hspace{3pt}}c@{\hspace{3pt}}@{\hspace{3pt}}c@{\hspace{3pt}}}
\hline
     &     & \multicolumn{3}{c}{procedure disabled} & \multicolumn{3}{c}{procedure enabled} \\
 $n$ & $p$ & run 1 & run 2 & run 3 & run 1 & run 2 & run 3 \\
\hline
100 & 0.1 & 58(593) & 58(1850) & 57(257) & \textbf{59(0.50)} & 59(8.47) & 59(24.6) \\
100 & 0.3 & 25(192) & 26(176) & 25(147) & 26(0.56) & 26(0.70) & \textbf{26(0.09)} \\
100 & 0.5 & 15(37.6) & 15(8.06) & 15(36.7) & 16(0.16) & 16(1.36) & \textbf{16(0.14)} \\
100 & 0.7 & 10(2.16) & 10(1.03) & 10(4.2) & 11(4.61) & \textbf{11(1.75)} & 11(2.08) \\
100 & 0.9 & 6(1.44) & 6(2.09) & 6(0.47) & 6(0.16) & \textbf{6(0.09)} & 6(0.14) \\
200 & 0.1 & 75(738) & 75(2352) & 75(3249) & 78(8.43) & \textbf{78(7.94)} & 78(9.69) \\
200 & 0.3 & 31(598) & 31(3052) & 32(3571) & 33(8.11) & 33(8.75) & \textbf{33(7.17)} \\
200 & 0.5 & 18(266) & 19(1486) & 19(169) & \textbf{20(3.19)} & 20(19.0) & 20(9.00) \\
200 & 0.7 & 12(235) & 11(44.4) & 12(307) & \textbf{12(0.23)} & 12(0.86) & 12(0.34) \\
200 & 0.9 & 7(2.37) & \textbf{7(2.12)} & 7(9.56) & 7(20.5) & 7(7.09) & 7(209)
\end{tabular}
\caption{Evaluation of \textsc{GenAlternativeSol}}  \label{tab:3}
\end{table}

\begin{table}[t] \centering \scriptsize
\begin{tabular}{@{\hspace{3pt}}c@{\hspace{3pt}}@{\hspace{3pt}}c@{\hspace{3pt}}@{\hspace{3pt}}c@{\hspace{3pt}}@{\hspace{3pt}}c@{\hspace{3pt}}@{\hspace{3pt}}c@{\hspace{3pt}}@{\hspace{3pt}}c@{\hspace{3pt}}@{\hspace{3pt}}c@{\hspace{3pt}}@{\hspace{3pt}}c@{\hspace{3pt}}@{\hspace{3pt}}c@{\hspace{3pt}}@{\hspace{3pt}}c@{\hspace{3pt}}@{\hspace{3pt}}c@{\hspace{3pt}}@{\hspace{3pt}}c@{\hspace{3pt}}@{\hspace{3pt}}c@{\hspace{3pt}}}
 & & & \multicolumn{5}{c}{$C=V$ (Grundy domination)} & \multicolumn{5}{c}{$C=\emptyset$ (Grundy total domination)} \\
 & & & \multicolumn{2}{c}{initial} & \multicolumn{2}{c}{final} & & \multicolumn{2}{c}{initial} & \multicolumn{2}{c}{final} & \\
Name & $n$ & dens (\%) & $UB$ & $LB$ & $UB$ & $LB$ & time & $UB$ & $LB$ & $UB$ & $LB$ & time \\
\hline
myciel3 & 11 & 36.36 & 6 & 6 &   &   & 0.0 & 9 & 8 & \textbf{8} & 8 & 0.7 \\
myciel4 & 23 & 28.06 & 17 & 13 & \textbf{15} & 13 & $-$ & 20 & 16 & \textbf{19} & 16 & $-$ \\
queen5\_5 & 25 & 53.33 & 8 & 7 & \textbf{7} & 7 & 63.3 & 9 & 8 & \textbf{8} & 8 & 288.2 \\
1-FullIns\_3 & 30 & 22.99 & 24 & 18$\dagger$ & 24 & 18 & $-$ & 26 & 18 & 26 & 18 & $-$ \\
queen6\_6 & 36 & 46.03 & 12 & 11$\dagger$ & 12 & 11 & $-$ & 13 & 11$\dagger$ & 13 & 11 & $-$ \\
2-Insertions\_3 & 37 & 10.81 & 32 & 28 & 32 & 28 & $-$ & 35 & 32 & 35 & 32 & $-$ \\
myciel5 & 47 & 21.83 & 40 & 27 & 40 & 27 & $-$ & 43 & 32 & 43 & 32 & $-$ \\
queen7\_7 & 49 & 40.48 & 19 & 14$\dagger$ & 19 & 14 & $-$ & 19 & 14 & 19 & 14 & $-$ \\
\hline
myciel3c & 11 & 63.64 & 4 & 4 &   &   & 0.0 & 4 & 4 &   &   & 0.0 \\
myciel4c & 23 & 71.94 & 5 & 5 &   &   & 0.0 & 5 & 5 &   &   & 0.0 \\
queen5\_5c & 25 & 46.67 & 10 & 9 & \textbf{9} & 9 & 266.6 & 10 & 8 & \textbf{8} & 8 & 1076 \\
1-FullIns\_3c & 30 & 77.01 & 8 & 8$\dagger$ &   &   & 0.0 & 8 & 7 & \textbf{7} & 7 & 15.3 \\
queen6\_6c & 36 & 53.97 & 10 & 9 & \textbf{9} & 9 & 7530 & 11 & 8 & \textbf{10} & 8 & $-$ \\
 &  &  &  &  &  &  &  & 9* & 8 & \textbf{8} & 8 & 8334 \\
2-Insertions\_3c & 37 & 89.19 & 4 & 4 &   &   & 0.0 & 4 & 4 &   &   & 0.0 \\
myciel5c & 47 & 78.17 & 8 & 6 & \textbf{6} & 6 & 80.5 & 8 & 6 & \textbf{6} & 6 & 112.5 \\
queen7\_7c & 49 & 59.52 & 10 & 9 & 10 & 9 & $-$ & 12 & 8 & 12 & 8 & $-$ \\
\hline
G25\_10 & 25 & 12.00 & 23 & 19 & 23 & 19 & $-$ & 25 & 22 & \textbf{23} & 22 & $-$ \\
G25\_30 & 25 & 31.00 & 15 & 13 & \textbf{13} & 13 & 8588 & 18 & 14 & \textbf{16} & 14 & $-$ \\
G25\_50 & 25 & 49.00 & 12 & 10$\dagger$ & \textbf{10} & 10 & 2983 & 15 & 12 & \textbf{12} & 12 & 2101 \\
G25\_70 & 25 & 71.33 & 8 & 7 & \textbf{7} & 7 & 17.0 & 9 & 8 & \textbf{8} & 8 & 61.0 \\
G25\_90 & 25 & 90.33 & 4 & 4 &   &   & 0.0 & 4 & 4 &   &   & 0.0 \\
G50\_10 & 50 & 9.31 & 47 & 38$\dagger$ & 47 & 38 & $-$ & 50 & 44$\dagger$ & 50 & 44 & $-$ \\
G50\_30 & 50 & 30.29 & 31 & 20$\dagger$ & 31 & 20 & $-$ & 33 & 24$\dagger$ & 33 & 24 & $-$ \\
G50\_50 & 50 & 51.76 & 17 & 13$\dagger$ & 17 & 13 & $-$ & 19 & 14 & 19 & 14 & $-$ \\
G50\_70 & 50 & 71.18 & 9 & 8 & \textbf{8} & 8 & 3625 & 11 & 10$\dagger$ & 11 & 10 & $-$ \\
G50\_90 & 50 & 89.14 & 5 & 5 &   &   & 0.1 & 6 & 6 &   &   & 0.1 \\
\hline
buenosaires\_full & 48 & 10.20 & 44 & 39$\dagger$ & 44 & 39 & $-$ &  &  &  &  &  \\
buenosaires\_small & 21 & 20.48 & 18 & 16$\dagger$ & \textbf{16} & 16 & 11952 &  &  &  &  &  \\
\end{tabular}
\caption{Evaluation on benchmark instances}  \label{tab:5}
\end{table}

\begin{table}[t] \centering \small
\begin{tabular}{@{\hspace{3pt}}c@{\hspace{3pt}}@{\hspace{3pt}}c@{\hspace{3pt}}@{\hspace{3pt}}c@{\hspace{3pt}}@{\hspace{3pt}}c@{\hspace{3pt}}@{\hspace{3pt}}c@{\hspace{3pt}}@{\hspace{3pt}}c@{\hspace{3pt}}@{\hspace{3pt}}c@{\hspace{3pt}}}
\hline
     &     &       & \multicolumn{2}{c}{upper bound} & \multicolumn{2}{c}{lower bound} \\
 $n$ & $r$ & $|V|$ & known & $m_3$ & known & Tabu \\
\hline
8 & 3 & 56 & 46 & \textbf{37} & 21 & \textbf{22}(0.03) \\
9 & 3 & 84 & 64 & \textbf{50} & 28 & 28(0.20) \\
10 & 3 & 120 & 85 & \textbf{65} & 36 & 36(1.00) \\
11 & 3 & 165 & 109 & \textbf{82} & 45 & 45(3.92) \\
12 & 3 & 220 & 136 & \textbf{101} & 55 & 55(13.2) \\
13 & 3 & 286 & 166 & \textbf{122} & 66 & 66(41.8) \\
14 & 3 & 364 & 199 & \textbf{145} & 78 & 78(107) \\
15 & 3 & 455 & 235 & \textbf{170} & 91 & 91(267) \\
16 & 3 & 560 & 274 & \textbf{197} & 105 & 105(624) \\
17 & 3 & 680 & 316 & \textbf{226} & 120 & 120(1382) \\
\hline
9 & 4 & 126 & 121 & \textbf{115} & 56 & \textbf{77}(0.44) \\
10 & 4 & 210 & 195 & \textbf{179} & 84 & \textbf{93}(4.08)$\dagger$ \\
11 & 4 & 330 & 295 & \textbf{265} & 120 & \textbf{121}(184)$\dagger$ \\
12 & 4 & 495 & 425 & \textbf{375} & 165 & 165(169) \\
13 & 4 & 715 & 589 & \textbf{512} & 220 & 220(861) \\
\hline
11 & 5 & 462 & 456 & \textbf{448} & 210 & \textbf{296}(66.3) \\
12 & 5 & 792 & 771 & \textbf{746} & 330 & \textbf{379}(2680)$\dagger$ \\
\end{tabular}
\caption{Improving known bounds for Kneser graphs}  \label{tab:4}
\end{table}

\begin{table}[t] \centering \small
\begin{tabular}{@{\hspace{3pt}}c@{\hspace{3pt}}@{\hspace{3pt}}c@{\hspace{3pt}}@{\hspace{3pt}}c@{\hspace{3pt}}@{\hspace{3pt}}c@{\hspace{3pt}}@{\hspace{3pt}}c@{\hspace{3pt}}@{\hspace{3pt}}c@{\hspace{3pt}}}
\hline
Graph & $C$ & $|V|$ & dens (\%) & $\grd(G;C)$ & time \\
\hline
$K_{22,4}$ & $\emptyset$ & 7315 & 41.8 & 70 & 912 \\
$K_{23,4}$ & $\emptyset$ & 8855 & 43.8 & 70 & 1649 \\
$K_{17,5}$ & $\emptyset$ & 6188 & 12.8 & 252 & 474 \\
$K_{18,5}$ & $\emptyset$ & 8568 & 15.0 & 252 & 1268 \\
$K_{15,6}$ & $\emptyset$ & 5005 & 1.7 & 924 & 222 \\
$K_{16,6}$ & $\emptyset$ & 8008 & 2.6 & 924 & 891 \\
$K_{15,7}$ & $\emptyset$ & 6435 & 0.12 & 3432 & 421 \\
$W_{10000}^{1000}$ & $V$ & 10000 & 20 & 8000 & 1856 \\
$W_{10000}^{2000}$ & $V$ & 10000 & 40 & 6000 & 2154 \\
$W_{10000}^{3000}$ & $V$ & 10000 & 60 & 4000 & 2398 \\
$W_{10000}^{4000}$ & $V$ & 10000 & 80 & 2000 & 2497
\end{tabular}
\caption{Performance of tabu search on large instances}  \label{tab:7}
\end{table}

\begin{figure}[t] \centering
\includegraphics[width=0.9\textwidth]{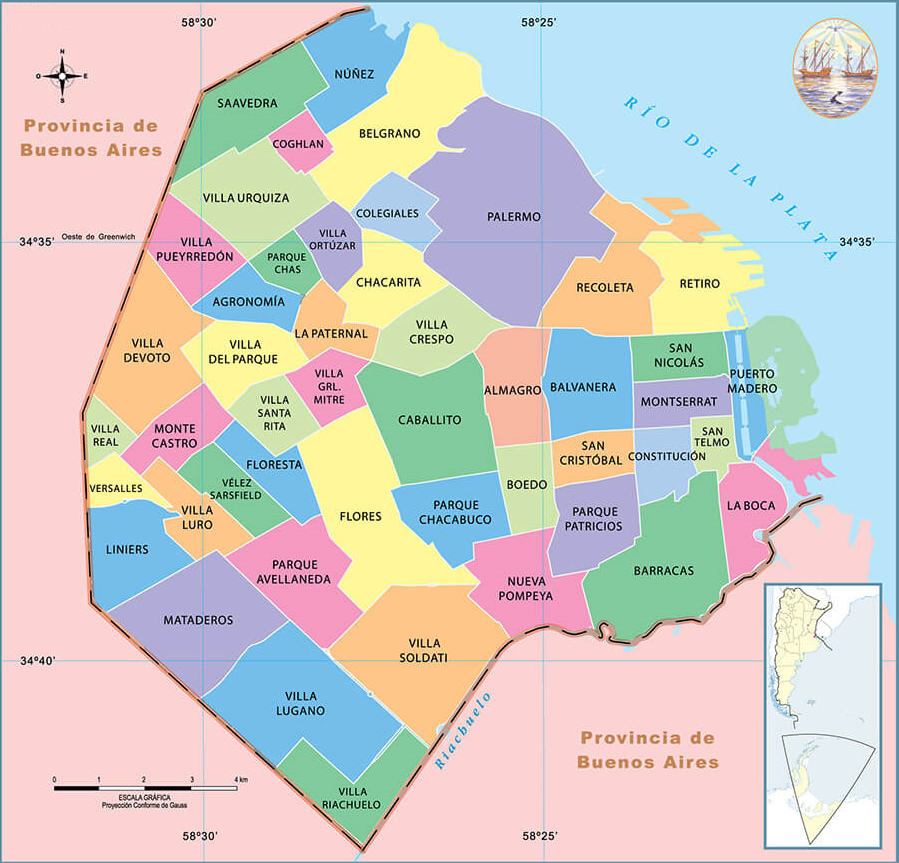}
\caption{Neighborhoods of the city of Buenos Aires}  \label{fig:bsas}
\end{figure}

\begin{table}[t] \centering \footnotesize
\begin{tabular}{@{\hspace{3pt}}c@{\hspace{3pt}}@{\hspace{3pt}}c@{\hspace{3pt}}@{\hspace{3pt}}c@{\hspace{3pt}}@{\hspace{3pt}}c@{\hspace{3pt}}@{\hspace{3pt}}c@{\hspace{3pt}}@{\hspace{3pt}}c@{\hspace{3pt}}}
\multicolumn{4}{c}{\emph{full version}} \\
\emph{Name} & \emph{Index} & \emph{Name} & \emph{Index} \\
\hline
Agronomía & 26 & Almagro & $-$ \\
Balvanera & 12 & Barracas & 6 \\
Belgrano & 14 & Boedo & 21 \\
Caballito & 27 & Chacarita & 17 \\
Coghlan & 3 & Colegiales & 15 \\
Constitución & 5 & Flores & 31 \\
Floresta & 33 & La Boca & 2 \\
La Paternal & $-$ & Liniers & 37 \\
Mataderos & 38 & Monte Castro & $-$ \\
Montserrat & 9 & Nueva Pompeya & 24 \\
Núñez & $-$ & Palermo & 16 \\
Parque Avellaneda & $-$ & Parque Chacabuco & 22 \\
Parque Chas & 20 & Parque Patricios & 7 \\
Puerto Madero & 10 & Recoleta & 13 \\
Retiro & 11 & Saavedra & 1 \\
San Cristóbal & 8 & San Nicolás & $-$ \\
San Telmo & 4 & Vélez Sársfield & 35 \\
Versalles & 36 & Villa Crespo & 18 \\
Villa del Parque & 29 & Villa Devoto & 32 \\
Villa General Mitre & 28 & Villa Lugano & $-$ \\
Villa Luro & $-$ & Villa Ortúzar & 19 \\
Villa Pueyrredón & 25 & Villa Real & 34 \\
Villa Riachuelo & 39 & Villa Santa Rita & 30 \\
Villa Soldati & $-$ & Villa Urquiza & 23 \\
\hline
 & \\
\multicolumn{4}{c}{\emph{small version}} \\
\multicolumn{3}{c}{\emph{Name}} & \emph{Index} \\
\hline
\multicolumn{3}{l}{(Agronomía, La Paternal, Parque Chas, Villa General Mitre,} & \\
\multicolumn{3}{r}{Villa Ortúzar, Villa del Parque, Villa Santa Rita)} & 13 \\
\multicolumn{3}{c}{(Almagro, Boedo)} & 7 \\
\multicolumn{3}{l}{(Balvanera, Constitución, Montserrat, San Cristóbal,} & \\
\multicolumn{3}{r}{San Nicolás, San Telmo)} & 4 \\
\multicolumn{3}{c}{(Barracas, La Boca, Parque Patricios)} & 5 \\
\multicolumn{3}{c}{Belgrano} & $-$ \\
\multicolumn{3}{c}{Caballito} & 9 \\
\multicolumn{3}{c}{(Chacarita, Colegiales, Villa Crespo)} & 8 \\
\multicolumn{3}{c}{(Coghlan, Núñez, Saavedra)} & 11 \\
\multicolumn{3}{c}{Flores} & 14 \\
\multicolumn{3}{l}{(Floresta, Liniers, Monte Castro, Vélez Sársfield,} & \\
\multicolumn{3}{r}{Versalles, Villa Luro, Villa Real)} & $-$ \\
\multicolumn{3}{c}{Mataderos} & 16 \\
\multicolumn{3}{c}{(Nueva Pompeya, Parque Chacabuco)} & 10 \\
\multicolumn{3}{c}{Palermo} & 6 \\
\multicolumn{3}{c}{Parque Avellaneda} & $-$ \\
\multicolumn{3}{c}{Puerto Madero} & 3 \\
\multicolumn{3}{c}{Recoleta} & 1 \\
\multicolumn{3}{c}{Retiro} & 2 \\
\multicolumn{3}{c}{Villa Devoto} & $-$ \\
\multicolumn{3}{c}{(Villa Lugano, Villa Riachuelo)} & $-$ \\
\multicolumn{3}{c}{Villa Soldati} & 15 \\
\multicolumn{3}{c}{(Villa Urquiza, Villa Pueyrredón)} & 12 \\
\hline
\end{tabular}
\caption{Best sequences for both versions (optimal for ``small'')}  \label{tab:6}
\end{table}

\end{document}